\documentclass{amsart}

\usepackage[english]{babel}
\usepackage{pgf,tikz,xypic,graphicx,mathrsfs,etex,bm,sidecap,hyperref}
\usepackage{tikz-cd}
\usepackage{stmaryrd} 
\xyoption{all}
\usetikzlibrary{snakes, shapes,arrows,patterns,decorations.pathreplacing,matrix}

\newtheorem{theorem}{Theorem}
\newtheorem{lemma}{Lemma}
\newtheorem{corollary}{Corollary}
\newtheorem{definition}{Definition}

\newcommand{\op}{\operatorname}

\newcommand{\tr}{\op{tr}}

\title[Isometries and conformal isometries]{Sachs equations and plane waves II: Isometries and conformal isometries}
\author{Jonathan Holland}
\address{Compunetix\\
  2420 Mosside Blvd \# 1\\
  Monroeville, PA 15146}
\author{George Sparling}
\address{University of Pittsburgh\\
  Department of Mathematics\\
  301 Thackeray Hall\\ Pittsburgh, PA 15260
}
\date{\today}

\begin{document}
\begin{abstract}
  This article describes the symmetries of plane wave spacetimes in dimension four and greater.  It begins with a description of the isometric automorphisms, and in particular the homogeneous plane waves.  Then the article turns to describing isometries from one plane wave to another.  The structure of the isometries is relevant for the problem of classifying vacuum spacetimes by observables, and the article presents an explicit example of a family of vacuum spacetimes encoding the Bernoulli shift, and so not classifiable by observables.  Next, the article turns to a description of the conformal isometries.  Here it is assumed that the conformal curvature does not vanish identically, the case of Minkowski space being both trivial and very degenerate.  The article then classifies all conformal automorphisms and isometries of plane waves of the Rosen and Brinkmann forms.
\end{abstract}
\maketitle
\tableofcontents
\section{Introduction}
In \cite{SEPI}, the authors introduced and described plane wave spacetimes.  We recollect some relevant definitions:
\begin{definition}
  A {\em Penrose limit} is a quadruple $(M,G,\mathcal D,\gamma)$ where:
  \begin{itemize}
  \item $M$ is a connected simply-connected manifold, and $G$ is a smooth metric tensor on $M$ of signature $(1,n+1)$;
  \item $\mathcal D_t:M\to M$ is a smooth one-parameter group of diffeomorphisms, acting properly on $M$, called the {\em dilation} group;
  \item $\mathcal D_t^*G=e^{2t}G$ for all $t\in\mathbb R$; and
  \item the limit $\lim_{t\to\infty}\mathcal D_t^{-1}x$ exists for all $x$ in $M$, and the limit points comprise a smooth curve $\gamma$ in $M$.
  \end{itemize}
  Moreover, the curve $\gamma$ is always a null geodesic in $M$, so it is called the {\em central null geodesic}.
\end{definition}

We then define a {\em plane wave} to be a Penrose limit spacetime where the $\mathcal D_t$ and $\gamma$ are ``forgotten''.  Thus a plane wave is a pair $(M,G)$ for which there exists a $\mathcal D_t:M\to M$ making $(M,G,\mathcal D_t,\op{Fix}(\mathcal D))$ a Penrose limit.  In general, every plane wave admits many dilations (at least a $(2n+1)$-dimensional family).  

The unstable manifolds of the dilation group $\mathcal D_t$ of a Penrose limit define the {\em wave fronts}.  In \cite{SEPI}, we show that the wave fronts are $n+1$ dimensional affine spaces on which the dilation acts as a linear transformations fixing an origin defined by the point where the central null geodesic intersects the wave front.  These spaces are ``unstable'' manifolds of $\mathcal D_t$ because they are the directions along which $\mathcal D_t$ expands.  We can thus form an intuitive picture of a Penrose limit as a null geodesic $\gamma$, with affine parameter $u$, and a foliation by wave fronts, one for each $\gamma(u)$ in $\gamma$.

The basic examples of Penrose limits are metrics on the space $\mathbb M = \mathbb U\times\mathbb R\times\mathbb X$ where $\mathbb U$ is a real interval and $\mathbb X$ a fixed real Euclidean space.  We use coordinates $(u,v,x)$ for $\mathbb M$ throughout, where $u\in\mathbb U, v\in\mathbb R$, and $x\in \mathbb X$.

\begin{itemize}
\item A {\em Brinkmann metric}:
  \begin{equation}\label{Brinkmann1}
2\,du\,dv + x^Tp(u)x\,du^2 - dx^Tdx
\end{equation}
where $u$ belongs to a real interval, $v\in\mathbb R$, $x$ belongs to a fixed real Euclidean $\mathbb X$ (which we can assume is $\mathbb R^n$ without loss of generality), and $p(u)$ is an $n\times n$ symmetric real matrix depending smoothly on $u$.
  Here the dilation group is
  \begin{equation}\label{BasicDilation}
    \mathcal D_t(u,v,x) = (u,e^{2t}v,e^tx).
  \end{equation}
  The central null geodesic is the curve $v=0,x=0$, $u$ arbitrary.  The wave fronts are the $u=$constant hypersurfaces.
\item A {\em Rosen universe}:
  \begin{equation}\label{RosenUniverse}
  \mathcal G_R(g) = 2\,du\,dv - dx^Tg(u)dx
\end{equation}
where $g(u)$ is a $n\times n$ positive-semidefinite matrix depending smoothly on $u\in\mathbb U$, a real interval.  (The matrix $g$ is positive definite away from a discrete set of points $\mathbb S\subset\mathbb U$, and the points of $\mathbb S$ where $g$ degenerates are removable singularities of the metric.)  The dilation is as in \eqref{BasicDilation}, with the same wave fronts and central null geodesic as in the Brinkmann case.
\item The {\em Alekseevsky metric}:
  \begin{align}
    \label{Alekseevsky1} \mathcal G_\alpha(p,\omega) &= du\,\alpha - dx^Tdx\\
    \notag \alpha &= 2\,dv + x^Tp(u)x\,du - 2x^T\omega(u)\,dx
  \end{align}
where $p(u),\omega(u)$ are smooth $n\times n$ matrices, $p=p^T,\omega=-\omega^T$.  (Again, the dilation, central null geodesic, and wave fronts are as in the other cases.)
\end{itemize}
An important result, ultimately due to Alekseevsky \cite{Alekseevsky}, but reformulated in this way in \cite{SEPI}, is:
\begin{theorem}\label{SEPITheorem1}
Every plane wave is a Brinkmann metric, a Rosen universe, and an Alekseevsky metric.
\end{theorem}

In all of these examples:
\begin{itemize}
\item the central null geodesic is the curve $\gamma : (u\in\mathbb U)\mapsto (u,0_{\mathbb R},0_{\mathbb X})$;
\item the dilation is the group $\mathcal D_t(u,v,x)=(u,e^{2t}v,e^tx)$; and
\item the wave fronts are the $u=$constant hypersurfaces.
\end{itemize}

Every plane wave has a group of at least $2n+1$ symmetries.  The Lie algebra of infinitesimal Killing symmetries typically annihilates $u$, and therefore the level sets of $u$ are invariant manifolds on which the Lie algebra integrates to a transitive group of transformations that all fix $u$.  In that case, the Killing algebra generates an isometry group that is transitive on the $2n+1$ dimensional manifold of null directions transverse to any fixed wave front.

In the exceptional case of a plane wave with a group of symmetries that does not fix the $u$ coordinate, the automorphism group is transitive.  They are called {\em homogeneous plane waves}.

The two main examples of homogeneous plane waves are as follows.  First, consider the Alekseevsky metric
\begin{align*}
 \mathcal G_A(\omega, p) &= du\,\alpha - dx^Tdx, \\
 \alpha &= 2\,dv - 2\,x^T\omega\,dx + x^Tpx\,du
\end{align*}
where $p,\omega$ are constant $n\times n$ matrices with $\omega^T=-\omega$, $p^T=p$.  Clearly these have $u\mapsto u+c$ for an extra symmetry, since nothing depends on $u$.  Secondly, consider the Brinkmann metric
$$\mathcal G_B(\omega,p) = 2\,du\,dv + x^Te^{-u\omega}pe^{u\omega}x\,du^2 - dx^Tdx.$$
(Note that $e^{u\omega}$ denotes the one parameter subgroup in $SO(n)$ generated by the constant skew matrix $\omega$.)  Now the extra symmetry is
$$u\mapsto u+c,\quad x\mapsto e^{-c\omega}x.$$
Every homogeneous plane wave is globally conformal to either of these forms, with conformal factor $e^{bu}$ for a constant $b$.  The obvious change of variables ($X = e^{u\omega}x$) is an isomorphism of Penrose limits.  We summarize:
\begin{lemma}\label{GAGBIso}
  Let $\omega=-\omega^T$ and $p=p^T$ be constant.  Then
  $$\mathcal G_B(\omega,p)\cong\mathcal G_A(\omega,p+\omega^2),$$
  and
  $$\mathcal G_A(\omega,p)\cong\mathcal G_B(\omega,p-\omega^2)$$
  where the isometric isomorphisms preserve the dilation and fix the central null geodesic pointwise.
\end{lemma}

This paper is organized as follows.  The first half of the paper is devoted to studying isometries of plane waves.  First, in \S\ref{AutomorphismsSection}, we characterize the automorphisms of Penrose limits.  Generically, these are only constant rotations in $\mathbb X$, but there is one more non-degenerate, non-flat case, in which there is a single extra automorphism.  Plane waves having an extra global automorphism are precisely the homogeneous plane waves.  The results of this section are already known, or easily extracted from the current literature on plane waves.  In four dimensions, the cases are taken up in \cite{Stephani}.  We refer also to \cite{blau2003homogeneous} which covers substantially similar ground, but it is presented here in a rather different way that we feel is justified to keep the presentation here as self-contained as possible, and because of its close connection with the problems taken up later in the article.  In \S\ref{HomogeneousPlaneWavesSection}, we discuss the case of the homogeneous plane waves.  These come in two flavors, depending on whether they are geodesically complete.  The geodesically complete plane waves are termed {\em microcosms}, and their detailed study shall be the subject of the next paper in this series.

Next, we consider the problem of isometries between two different Brinkmann metrics.  Modulo the symmetries, these turn out to be pretty much as one expects: the freedom is an affine reparameterization of the $u$ coordinate and a constant rotation of $\mathbb X$.  This rather simple result is used to construct an example of a family of vacuum plane waves in four dimensions, continuously parameterized by the Hilbert cube $[0,1/2]^{\mathbb Z}$, such that two elements of the family are isometric if and only if they are shift-equivalent.  This example has implications for the complexity of classifying vacuum spacetimes, along the lines of \cite{panagiotopoulos2023incompleteness}.

The second part of the paper is concerned with the corresponding theory for conformal isometries, which ostensibly have a more intricate structure.  In \S\ref{ConformalDiffeoRosen}, we characterize the conformal isometries between Rosen plane waves.  In \S\ref{BrinkmannConformalSymmetries}, we find the conformal Killing algebra of an arbitrary plane wave, under the assumption that the conformal curvature be nonzero at some point.  The main result of the paper is Lemma \ref{FundamentalLemma}, although the statement is not very pretty, so the first part of \S\ref{BrinkmannConformalSymmetries} explores some corollaries and other consequences.  We can here summarize the main result as follows:
\begin{theorem}
  Let $(\mathbb M, G)$ be a conformal plane wave whose conformal curvature is non-vanishing at some point.  Then the conformal symmetry algebra $\mathcal W$ is spanned by the algebra of infinitesimal Killing symmetries plus a dilation vector field vanishing on a null geodesic.  The first derived algebra $\mathcal W^{(1)}=[\mathcal W,\mathcal W]$ is spanned by the Heisenberg symmetries, plus the subalgebra of infinitesimal rotations in $\mathfrak{so}(\mathbb X)$ that commute with the tidal curvature matrix $p(u)$ for all $u$.  The center of $\mathcal W^{(1)}$ is one-dimensional, spanned by the null Killing symmetry $\partial_v$.  Finally, exactly one of the following statements is true:
  \begin{itemize}
  \item $\mathcal W/\mathcal W^{(1)}$ is one-dimensional, containing only the equivalence class of the dilation, in which case every infinitesimal conformal symmetry is complete.
  \item $\mathcal W/\mathcal W^{(1)}$ is two-dimensional, in which case it contains the dilation and an additional infinitesimal Killing symmetry, and $\mathcal W$ is locally transitive.  In that case, the plane wave is canonically embedded into a homogeneous conformal plane wave, with Alekseevsky metric $\mathcal G_A(\omega,p)$, where $\omega^T=-\omega$, $p^T=p$, and $p,\omega$ are constant.  The group of conformal symmetries of the ambient plane wave is transitive on the set of all null vectors transverse to some wave front.
  \end{itemize}
  In both cases, the group of conformal symmetries of $(\mathbb M,G)$ is transitive on any (fixed) wave front, as well as on the set of null vectors transverse to that wave front.
\end{theorem}

\subsection{Definitions}
\begin{definition}
An Euclidean vector space is a pair $(\mathbb{X}, ^T)$ with  $\mathbb{X}$, a finite-dimensional vector space over the reals, $\mathbb{R}$, and $^T$ an idempotent transpose operation on the tensor algebra of $\mathbb{X}$ over the reals, such that for any tensors $A$ and $B$ we have $(A\otimes B)^T = B^T \otimes A^T$.  Also $^T$ commutes with tensor contraction.  Finally $^T$ interchanges $\mathbb{X}$ and $\mathbb{X}^*$, the dual space of $\mathbb{X}$, such that if $0 \ne x \in \mathbb{X}$, then $x^Tx > 0$.  \end{definition}
An endomorphism $L$ of $\mathbb{X}$, so $L \in \mathbb{X} \otimes \mathbb{X}^*$ is said to be symmetric if and only if $L^T = L$, or to be skew if and only if $L^T = - L$.  Denote by $I$ the (symmetric) identity automorphism of $\mathbb{X}$.

No essential loss of generality is entailed by taking $\mathbb X$ to be the standard Euclidean space $\mathbb R^n$, with the transpose operation the usual matrix transpose, and call an endomorphism of $\mathbb X$ an ($n\times n$) {\em matrix}.

We now have the various models of plane wave metrics:
\begin{definition}
  Let $\mathbb M=\mathbb U\times\mathbb R\times\mathbb X$ where $\mathbb U$ is a real interval, and $\mathbb X$ is a Eudlidean space of dimension $n$.  On $\mathbb M$, define the various types of plane wave metrics by:
  \begin{itemize}
  \item A Brinkmann metric is the Lorentzian metric
    $$G_\beta(p) = 2\,du\,dv + x^Tp(u)x\,du^2 - dx^Tdx$$
    where $p$ is a symmetric matrix, which depends smoothly on $u\in\mathbb U$.
  \item A Rosen metric is the Lorentzian metric
    $$ G_\rho(h) = 2\,du\,dv - dx^Th(u)dx$$
    where $h$ is a symmetric positive-definite matrix depending smoothly on $u\in\mathbb U$.
  \item An Alekseevsky metric is a metric of the form 
    \begin{align*}
      \mathcal G_A(\omega, p) &= du\,\alpha - dx^Tdx, \\
      \alpha &= 2\,dv - 2\,x^T\omega(u)\,dx + x^Tp(u)x\,du
    \end{align*}
    where $\omega$ is skew, and $p$ is symmetric, and both depend smoothly on $u\in\mathbb U$.
  \end{itemize}
\end{definition}
In \cite{SEPI}, it was shown how to go from a Rosen or Alekseevsky metric to a Brinkmann metric (globally), and also from a Brinkmann metric to a corresponding Rosen metric {\em locally} in $u$.

The change between a Brinkmann metric $G_\beta(p)$ and Rosen metric $G_\rho(\bar h)$ is by means of the coordinate change $x = L\bar x$, $v = \bar v + 2^{-1}x^TSx$, where $L$ is a matrix satisfying $\bar h=L^TL$ and $\dot L = SL$, where the symmetric tensor $S$ is a solution of the Sachs equation $\dot S+S^2+p=0$.  As usual, a dot denotes differentiation with respect to $u$.  This type of matrix has the property that its columns form a Lagrangian basis of Jacobi fields, so we make the definition:
\begin{definition}
  A {\em Lagrangian matrix} is a matrix $L(u)$ depending smoothly on $u$, such that $\ddot L + pL=0$, $L(u)$ is invertible for all $u\in\mathbb U-\mathbb S$, where $\mathbb S$ is a discrete set of points, and $S(u):=\dot L(u)L(u)^{-1}$ is symmetric for all $u\in\mathbb U-\mathbb S$.
\end{definition}

Many of the results of this paper assume that the plane waves under consideration have nontrivial curvature somewhere.  This condition is most easily introduced for the Brinkmann form of a plane wave, $G_\beta(p)$.  The quantity $p(u)$ is smooth on $\mathbb{U}$, taking values in the real symmetric endomorphisms of $\mathbb{X}$.  We decompose $p(u) = \tilde{p}(u) + P(u)I$, where $nP(u) = \op{tr}(p(u))$ and $\op{tr}(\tilde{p}(u)) = 0$.    Recall that the real-valued function $P$ determines the Einstein tensor of the spacetime, such that positive energy density corresponds to $P$ being positive, whereas  the trace-free part of $p(u)$, denoted $\tilde{p}(u)$,  a trace-free symmetric  endomorphism of $\mathbb{X}$, determines the Weyl curvature.
\begin{definition}
  A spacetime is called non-flat if the curvature is not identically zero, and flat otherwise.  A spacetime is said to be conformally curved if its Cartan conformal curvature is not identically zero and conformally trivial otherwise.  
\end{definition}
  Any spacetime conformally isometric  to a conformally curved spacetime is conformally curved.   Any spacetime conformally isometric to a spacetime that is not conformally trivial is itself conformally trivial.

By standard Bianchi identities, if the spacetime dimension is at least four, then the spacetime is conformally curved if and only if its Weyl curvature is not identically zero.   For a plane wave of spacetime dimension three (so $n = 1$), it is easy to see that it is always at locally conformally flat so it is conformally trivial. 

The purpose of this article is to study isometries and conformal isometries:
\begin{definition}
Let $(M,G)$ and $(\bar M,\bar G)$ be spacetimes.  A smooth diffeomorphism $\phi : M\to\bar{M}$ is called an {\em isometry} if $\phi^*\bar G=G$.  The diffeomorphism is called a conformal isometry if $\phi^*G = e^{2\Omega}G$ for some smooth function $\Omega$ on $M$.
\end{definition}
In this definition, and elsewhere unless otherwise explicitly indicated, diffeomorphisms are assumed to be {\em bijective}, smooth and with smooth inverse.

Infinitesimal isometries, or Killing vectors, are vector smooth fields $V$ such that
$$\mathscr L_VG = 0,$$
and infinitesimal conformal diffeomorphisms, or conformal Killing vectors, are smooth vector fields $V$ such that
$$\mathscr L_VG = SG $$
where $S$ is a smooth function on $M$.

This article studies several related questions:
\begin{itemize}
\item Determine the isometries (resp., the conformal isometries) of a plane wave to itself (the isometry group, resp., the conformal group).
\item Determine the isometries (resp., the conformal isometries) of a plane wave in either Brinkmann or Rosen form to another plane wave of the same form.
\end{itemize}

The common method for solving these problems is first to use the standard fact, established in \cite{SEPI}, that there is a subgroup of the isometry group of any plane wave that acts transitively on the set of transverse null geodesics.

\begin{definition}
A null geodesic in a plane wave is called {\em transverse} if it intersects any (and hence every) wave front transversally.
\end{definition}

In \cite{SEPI}, we proved that the isometry group of a plane wave is transitive on the set of transverse null geodesics.  It is also true for conformal isometries.  Therefore, after showing that isometries preserve transversality (for suitably generic plane waves), to characterize (conformal) isometries between plane waves, it is sufficient to characterize (conformal) isometries that send the central null geodesic of one to the central null geodesic of the other.

\section{Isometries}
\subsection{Automorphisms of Penrose limits}\label{AutomorphismsSection}
Consider a plane wave spacetime in the Brinkmann form
$$G=2\,du\,dv + x^Tp(u)x\,du^2 - dx^Tdx.$$
Notice that the vector field $V=\partial_v$ is a Killing vector of $G$.  Also the one-parameter group of dilations $\mathcal D_t:(u,v,x)\mapsto (u,e^{2t}v,e^tx)$ preserves the null geodesic $\gamma=\{v=0, x=0\}$ is a null geodesic, which is affinely parameterized by the coordinate $u$.

Recall that a plane wave equipped with such a dilation is called a {\em Penrose limit} spacetime.
\begin{definition}
  An {\em infinitesimal automorphism} of a Penrose limit is a Killing vector that commutes with the dilation.
\end{definition}

\begin{lemma}
The dilation of a Penrose limit is tangent to a foliation by null hypersurfaces. The vector field $V=\partial_v$ is the unique null Killing vector field tangent to the null hypersurfaces which is normalized against the tangent vector $\partial_u$ to $\gamma$.
\end{lemma}
\begin{proof}
  The vector field in question is obviously null and tangent to the $du=0$ hypersurfaces.  Therefore it is the unique null generator of these hypersurfaces, and so any Killing vector field is a multiple of $V$.  It is clear from the Killing equation that the only freedom is replacing $V$ by a constant multiple.
\end{proof}
Thus a Penrose limit spacetime has not only a preferred dilation (by definition), but also a preferred null Killing field $V$ up to an overall constant scale.  Therefore the following definition is natural:

\begin{definition}
  Call a Killing vector {\em commutant} if it commutes with the null Killing field $V=\partial_v$.
\end{definition}

\begin{lemma}\label{CommutantAutomorphisms}
  Suppose that $W$ is an infinitesimal commutant automorphism that vanishes at a point of $\gamma$.  Then $W$ is a vector field
  $$W = \mathbf W(x)\partial_x $$
  where $\mathbf W$ is a constant element of $\mathfrak{so}(\mathbb X)$ such that $p(u)\mathbf W$ is skew for all $u\in\mathbb U$.
\end{lemma}
(In particular, generically the infinitesimal commutant automorphisms are reduced to the identity.)
\begin{proof} Let
  $$W = A\partial_u + B\partial_v + C\cdot\partial_x $$
  and 
  $$\mathcal D=2v\,\partial_v + x\partial_x$$
  be the dilation.  Since $W$ commutes with $\mathcal D$ by assumption, $A$ is a function only of $u$.
  
  Now, $B$ has has weight two and $C$ has weight one.  Therefore,
  $$B = b_1(u)v + x^Tb_2(u)x, $$
  $$C = c(u)x.$$
  However, $b_1=0$ because $W$ commutes with $\partial_v$ by hypothesis.

  We now examine the Killing equations $\mathscr L_W G = 0$, where $G=2\,du\,dv + x^Tp(u)x\,du^2 - dx^Tdx$.  Because of the sole $du\,dv$ term in the LHS, we have $\dot A=0$, and by the hypothesis that $W$ vanishes at a point of $\gamma$, $A$ is identically zero.  The only term quadratic in $dx$ has a factor of $c+c^T$, so this too must be zero, and $c(u)$ belongs to $\mathfrak{so}(\mathbb X)$ for all $u$.  The $du\,dx$ term is $x^T(2\,b_2-\dot c)\,dx\,du=0$ so $b_2=0$ and $c$ is constant, because $b_2$ is symmetric but $\dot c$ is skew.  Finally, the only $du^2$ term is $x^T(c^Tp+pc)x$, so $pc$ is skew.
  
\end{proof}

Generically, the only automorphisms of Penrose limits are those just described.  In some cases however, there is an extra automorphism:

\begin{theorem}\label{Dichotomy}
  For a Penrose limit spacetime $G=2\,du\,dv + x^TP(u)x\,du^2 - dx^Tdx$, such that $P(u)$ does not vanish identically, exactly one of the following situations holds:
  \begin{itemize}
  \item Every infinitesimal automorphism is commutant and fixes the central null geodesic pointwise.
  \item There exists at most one independent infinitesimal automorphism that does not fix the central null geodesic,
    $T_{a,b,C}=(a-bu)\partial_u+bv\partial_v + Cx\cdot\partial_x$ where $a,b$ are real constants $a,b$ (not both zero), and $C$ is a real constant skew-symmetric matrix.
    This extra symmetry is present if and only if the parameters $a,b,C$ satisfy the following equation:
    \begin{equation}\label{ExtraSymmetryEqn}
      (a-bu)\dot P - 2bP + [P,C] = 0.
    \end{equation}
  \end{itemize}
\end{theorem}
Note that in the case of flat space, there are two such extra automorphisms, since we can take $C=0$ and there is no restriction on the parameters $a,b$. 
\begin{proof}
  Consider an automorphism
  $$T = \mathcal A\partial_u + \mathcal B\partial_v + \mathcal C(\partial_x).$$
  Because it commutes with the dilations, $\mathcal A$ has weight zero (so is a function of $u$), $\mathcal B$ has weight two, and $\mathcal C$ has weight $1$.  Thus $\mathcal A=a(u), \mathcal B=b(u)v + x^TB(u)x$ with $B$ symmetric, $\mathcal C=C(u)x$.  A priori, the functions depend on $u$ only.

  The metric is
  $$G = 2\,du\,dv + x^TP(u)x\,du^2 - dx^Tdx.$$
  The Killing equation is $\mathscr L_TG=0$.  We analyze the components of the equation.
  \begin{itemize}
  \item $du\,dv$:  The term is $2\,(\dot a+b)\,du\,dv$.  Thus we have $\dot a=-b$.
  \item $dx\,dx$: $C+C^T=0$ ($C$ is skew)
  \item $du\,dx$: $B - \dot C=0$.  But $\dot C$ is skew and $B$ is symmetric, so $B=0$ and $C$ is constant.
  \item $du^2$: $2 \dot bv + a x^T\dot Px + 2 \dot a x^TPx + x^T(C^TP + PC)x $.  We conclude that $\dot b=0$, so that $a=a_0 - bu$ and $(a_0-bu)\dot P - 2bP + [P,C]= 0$.
  \item $dv^2$:  None
  \item $dv\,dx$: None
  \end{itemize}
  The $du^2$ term thus implies the dichotomy stated by the theorem.

  Now, suppose that we have a pair of solutions $(a,b,C)$ and $(a',b',C')$ to \eqref{ExtraSymmetryEqn}, such that $(a,b)$ and $(a',b')$ are linearly independent vectors in $\mathbb R^2$.  Without loss of generality, we may assume $(a,b)=(1,0)$ and $(a',b')=(0,1)$.  Then, solving \eqref{ExtraSymmetryEqn} for each respective solution we have, on the one hand
  $$ P = e^{C u}P_0 e^{-Cu},$$
  and on the other hand, in an open interval around $u_0$ (assuming, as we may, that $u_0\not=0$)
  $$ P = e^{C' \ln u}Q_0e^{-C'\ln u}.$$
  Therefore, $P=P_0=Q_0$, and $P$ commutes with $C$ and $C'$.  But \eqref{ExtraSymmetryEqn} now gives $-2P=[C',P]=0$,
  so $P=0$, and the metric is flat.
\end{proof}

In the exceptional case, we have an equation
$$(a-bu)\dot P - 2bP = [C,P].$$
Further specialization to the case $b=0,a\not=0$, this equation evidently gives the evolution of $P$ under the action of a one-parameter group of rotations, so that $P=e^{uC}P_0e^{-uC}$.  It is easy to see that the case of $b\not=0$ can be reduced to that of $b=0$ by reparameterizing the $u$ coordinate and simultaneously rescaling $P$:
$$d\widetilde u =\frac{du}{a-bu}, \quad \widetilde P = (a-bu)^2P.$$
Indeed, we find then:
$$\frac{d\widetilde P}{d\widetilde u} = [C,\widetilde P].$$
So, replacing $\widetilde u$ by $t$, we have proven:
\begin{theorem}\label{ExtraSymmetryP}
  For a special plane wave with extra symmetry $T_{a,b,\rho}$, the tidal curvature is
  $$P = \left(\frac{dt}{du}\right)^2e^{\rho t}P_0e^{-\rho t} $$
  where $P_0$ is constant and $dt = \frac{du}{a-bu}$.
\end{theorem}

As an example, supposes that $\rho=0$ in the theorem.  Then $P$ satisfies $(a-bu)\dot P-2bP=0$, so that
$$P = (a-bu)^{-2}P_0.$$
The Brinkmann metric is then
$$G = 2\,du\,dv + (a-bu)^{-2} x^TP_0X\, du^2 - dx^Tdx.$$
The extra symmetry is
$$T_{a,b} = (a-bu)\partial_u+bv\partial_v,$$
and we have $\mathscr L_{T_{a,b}}G=0$.

Note that the metric $G$ has an essential singularity if $P_0\not=0$ and $b\not=0$, at $u=a/b$, because the tidal curvature of the central null geodesic blows up there.  In particular, it is not geodesically complete.

\subsection{Homogeneous plane waves}\label{HomogeneousPlaneWavesSection}
As noted in the previous section, some Penrose limits have an extra infinitesimal symmetry.  For example, if $P$ and $\Omega$ are constant, then the Alekseevsky metric
\begin{equation}\label{MicrocosmAlex}
  \mathcal A(P,\Omega) = 2\,du\,dv + x^TPx\,du^2 + [x^T\Omega\,dx -dx^T\Omega x]\,du - dx^Tdx
\end{equation}
has $u\to u+u_0$ as an extra symmetry.  Thus the symmetry group is transitive on each wave front, and transitive on the set of wave fronts, and so is transitive on the spacetime.  It is even transitive on the set of null directions transverse to some wavefront (i.e., as points in the projective null cone bundle, a $2n+2$ dimensional space).

Alternatively, for constant $P=P^T$ and $\rho=-\rho^T$, the Brinkmann metric
\begin{equation}\label{MicrocosmBrink}
\mathcal A_\beta(P,\rho) = 2\,du\,dv + x^Te^{\rho u}Pe^{-\rho u}x\,du^2 + dx^Tdx
\end{equation}
has for extra symmetry $u\to u+u_0$, $x\to e^{-\rho u_0}x$.
\begin{lemma}\label{MicrocosmAB}
  Every Penrose limit of the form \eqref{MicrocosmAlex} is globally isomorphic to one of the form \eqref{MicrocosmBrink}, and vice versa.  
\end{lemma}
\begin{proof}
  Beginning with
  $$G = 2\,du\,dv + x^Te^{\rho u}P_1e^{-\rho u} x\,du^2 - dx^Tdx,$$
  put $x=e^{\rho u}X$,
  $$2\,du\,dv + X^T(P_1+\rho^2)X\,du^2 + (X^T\rho dX - dX^T\rho X)du  - dX^TdX $$
  which is \eqref{MicrocosmAlex} with $P\to P_1+\rho^2$, $\Omega=\rho$.  This isomorphism is clearly reversible.
\end{proof}

\begin{theorem}\label{MicrocosmTheorem}
  Suppose that a Penrose limit has an extra infinitesimal automorphism $T_{a,b,\rho}$.  Then there is a conformal isometry of each component of $a-bu\not=0$ in $\mathbb U$ onto an open subset of $\mathcal A(P-b^2/4,\Omega)$, with conformal factor $|a-bu|$.  
\end{theorem}
\begin{proof}
  Start with the metric
  $$G=2\,du\,dv + x^TP(u)x\,du - dx^Tdx.$$
  Changing signs of both $a,b$ if necessary, we can assume $a-bu>0$.  Because of the extra symmetry,
  $$P(u) = \left(\frac{dt}{du}\right)^2e^{\rho t}P_0e^{-\rho t} $$
  with $(a-bu)dt=du$.  So
  $$G=2\,du\,dv + x^Te^{\rho t}P_0e^{-\rho t}x\,dt^2 - dx^Tdx.$$
  With $x=(a-bu)^{1/2}X$, $v=V - \tfrac{b}{4}X^2$,
  \begin{align*}
    G&=2(a-bu)\,dt\,dv + x^Te^{\rho t}P_0e^{-\rho t}x\,dt^2 - dx^Tdx\\
     &=(a-bu)\left(2\,dt\,dV + X^Te^{\rho t}\left[P_0-\tfrac{b^2}{4}I\right]e^{-\rho t} X\,dt^2 - dX^TdX\right).
  \end{align*}
  By Lemma \ref{MicrocosmAB}, we are done.
\end{proof}

\begin{definition}
A plane wave is homogeneous if its isometry group is transitive on points.
\end{definition}

Noting that the $t$ variable in the above proof can be chosen so that (modifying signs of $b$ if necessary), $e^{-bt}=a-bu$, we obtain the following Corollary (where $t$ has been replaced by a new $u$):
\begin{corollary}
  A homogeneous plane wave is isometric to a metric of the form
  \begin{equation}\label{GeneralHomogeneous}
    e^{-bu}\left(2\,du\,dv + x^Te^{\rho u}Pe^{-\rho u}x - dx^Tdx\right).
  \end{equation}
  Moreover, every locally homogeneous plane wave isometrically embeds into a unique metric of the form \eqref{GeneralHomogeneous}.
\end{corollary}
Here the global symmetry is $u\to u+c$, $x\to e^{(2^{-1}b+\rho)c} x$.

Note that, for $b\not=0$, the metric is not geodesically complete, because the affine parameter $e^{-bu}du$ does not parameterize the whole real line.  This does not bother us: since the conformal factor is rather trivial, it can be dealt with later if desired.  But from our perspective, it is the conformal structure that is primarily of interest.

Putting together this Lemma, and the proof of Theorem \ref{MicrocosmTheorem}:
\begin{theorem}
  Every homogeneous plane wave is globally conformal to a metric of the form \eqref{MicrocosmAlex}.
\end{theorem}

\begin{definition}
  A geodesically complete homogeneous plane wave is called a {\em microcosm}.
\end{definition}

The following combination of observations made thus far gives a convenient shorthand for thinking concretely about microcosms.
\begin{lemma}
  Every microcosm has a global metric of the form
  $$\mathcal G_A(p,\omega) = du\,\alpha -dx^Tdx, \qquad \alpha=2\,dv + x^Tpx\,du - 2x^T\omega\,dx$$
  where $p=p^T,\omega=-\omega^T$ are real constant $n\times n$ matrices.
\end{lemma}

\subsection{Brinkmann isometries}\label{DiffeomorphismsBrinkmann}
Given two Brinkmann plane wave metrics
\begin{align*}
  G &= 2\,du\,dv + x^Tp(u)x\,du^2 - dx^Tdx\\
  \bar G&= 2\,d{\bar u}\,d{\bar v} + {\bar x}^T{\bar p}({\bar u}){\bar x}\,d{\bar u}^2 - d{\bar x}^Td{\bar x}
\end{align*}
on, respectively, the manifolds $\mathbb M=\mathbb U\times\mathbb R\times\mathbb X$ and $\bar{\mathbb M}=\bar{\mathbb U}\times\mathbb R\times\mathbb X$, we determine the conditions under which there exists a diffeomorphism $\phi:\mathbb M\to\bar{\mathbb M}$ that is an isometry of $G$ onto $\bar G$.

\begin{theorem}\label{BrinkmannDiffeos}
  The Brinkmann metrics $G,\bar G$ are isometric if and only if $\bar{\mathbb U}$ is an image of $\mathbb U$ under an affine transformation ${\bar u}=a(u-u_0)$, where $a$ is a nonzero real constant, and there exists a constant $\gamma\in\op{O}(n)$ such that
  \begin{equation}\label{gamma}
    a^2\gamma^T{\bar p}({\bar u})\gamma = p(u).
  \end{equation}
  In that case, an isometry between the metrics is:
  \[{\bar x} = \gamma x, \qquad {\bar v}=v/a.\]
(The general isometry between the spacetimes is the composite of one of this form with an automorphism of $G$.)
\end{theorem}
\begin{proof}
  The converse direction is obvious: the given transformation is easily seen to be the claimed isometry.

  Also, we see immediately that if $G$ and $\bar G$ are diffeomorphic and flat then the functions $p$ and ${\bar p}$ are both identically zero.   Conversely if either $p$ or ${\bar p}$ is identically zero, then either the metrics are not diffeomorphic or both $p$ and ${\bar p}$ are identically zero.

  For the other direction, suppose that there is an isometry $\phi:\mathbb M\to\bar{\mathbb M}$, such that $\phi^*\bar G=G$, and $\mathbb M$ is not flat.  Let $\gamma$ and $\bar\gamma$ be the respective central null geodesics in $\mathbb M$ and $\bar{\mathbb M}$, and $\mathcal D,\bar{\mathcal D}$ the respective standard dilations.  Then $\phi\gamma$ is a null geodesic in $\bar{\mathbb M}$, fixed by the pushforward dilation $\phi_*\mathcal D$.  By Lemma \ref{DilationOfDvIffFlat} below, $\phi\gamma$ is transverse to the wave fronts of the dilation $\bar{\mathcal D}$.  By applying a symmetry of $\bar{\mathbb M}$, we can therefore map $\phi\gamma$ to the central null geodesic.

Because the isometry must send orbits of the symmetry group of $G$ to orbits of the symmetry group of $\bar G$, it follows that ${\bar u}$ is a function of $u$ only.  Moreover, ${\bar u}$ can only be an affine function of $u$, because it is an affine parameter of the central null geodesic.  For \eqref{gamma}, the isometry must send the quadratic form $dx^Tdx$ to $d{\bar x}^Td{\bar x}$ on the $u$ and ${\bar u}$ constant hypersurfaces, so that ${\bar x}$ is a rotation of $x$ for each $u$.  Because the vector field $\partial_x$ is parallel-transported up the central null geodesic, the additional freedom is precisely a constant rotation $\gamma$ of the coordinates $x$.

Finally, because the center of the first derived algebra of Killing symmetries of $G$ must match that of $\bar G$, so the vector fields $\partial_v$ and $\partial_{{\bar v}}$ must agree up to a constant:
\[ \partial_{{\bar v}} = a \partial_v.\]
Here $a$ must be the same constant as \eqref{gamma}, because $du\,dv=d{\bar u}\,d{\bar v}$.
\end{proof}

\begin{lemma}\label{DilationOfDvIffFlat}
  There exists a dilation fixing the null geodesic $u=0$, $x=0$ in the Brinkmann spacetime $(\mathbb M,G)$, if and only if $\mathbb M$ is (globally) Minkowski space ($\mathbb R^{n+2}$) and $G$ is the standard flat metric $2\,du\,dv - dx^Tdx$.
\end{lemma}
\begin{proof}
  Suppose that $D$ is a dilation fixing the null geodesic in question, which has tangent vector $\partial_v$.  The proof is basically an application of Lemma \ref{FundamentalLemma}, which would proceed along the following lines. By that lemma, we would have (under the additional assumption that $G$ is conformally curved)
  $$D = w\partial_u + 2^{-1}[\dot wx(\partial_x) + (Wx)(\partial_x)]$$
  where $W$ is a constant skew-symmetric matrix, $\dot w=2$,
  \begin{equation}\label{lemma6fundamental}
  \begin{array}{rl}
    0 &= -4\tr p\dot w -2\tr\dot pw\\
    0 &= 4\tilde p\dot w+2\dot{\tilde p} w + \tilde pW-W\tilde p.
  \end{array}
\end{equation}
Assuming without loss of generality that $w=2u$, we have
  $$8\tr p + 2\tr\dot pw = 0 \implies \tr p = \frac{C}{u^2}$$
  which must exist for all $u$, and therefore $\tr p=0$.

  The second equation gives
  $$8\tilde p + 4u\dot{\tilde p} + \tilde pW-W\tilde p=0.$$
  At $u=0$, this equation is the linear equation on $\tilde p(0)$:
  $$8\tilde p(0) + \tilde p(0)W-W\tilde p(0).$$
  Consider the linear operator $B(X) = 8X + XW-WX$ acting on trace-free symmetric matrices $X$, where $W$ is fixed and skew.  Then $B$ is a non-zero multiple of the identity operator plus a skew operator.  Therefore, its eigenvalues all have the same nonzero real part, and a fortiori are all nonzero.  So the kernel of $B$ is zero, and therefore $\tilde p(0) = 0$.  Therefore $p(0)=0$.  But, by Theorem \ref{ExtraSymmetryP}, if $p$ vanishes at any point of a homogeneous plane wave, then it vanishes at every point, and the metric is flat.

  This argument does not quite hold up though, because of the extra assumption that $G$ is conformally curved, so we must show that the conclusion of Lemma \ref{FundamentalLemma} still holds for {\em dilations} in the conformally flat case.  That is, if
  $$D = A\partial_u+B\partial_v +\partial_x\cdot C$$
  satisfies $\mathscr L_DG=2G$, then we have to show that \eqref{lemma6fundamental} still holds.  There were two places in the proof of Lemma \ref{FundamentalLemma} where the non-vanishing of $\tilde p$ was invoked: once to show that $q=0$ and once to show that $y=0$.  We claim that these are both automatic in the present case.  In Lemma \ref{FundamentalLemma}, we had $S = 2\,x^T\dot y + \dot w +r_0 + 2q(u)v$, but in the present case $S=2$.  Therefore $\dot y=0$, $q=0$.  On the other hand, $y$ also satisfies the Jacobi equation $\ddot y+py=0$.  Hence, because $y$ is constant and $p$ is pure trace, either $y=0$ or $p=0$.  The latter case is the content of the present Lemma, and in the former case, the argument of the first two paragraphs of the proof goes through.
\end{proof}

\begin{lemma}
Any two dilations of a plane wave fixing the same null geodesic pointwise have the same unstable manifolds.
\end{lemma}
\begin{proof}
  Using a Brinkmann metric $G$, suppose that a geodesic $\gamma$ is fixed by the standard dilation $D=2v\partial_v+x\partial_x$.  Then any other dilation
  $$D' = a\partial_u + b\partial_v + c\partial_x$$
  fixing $\gamma$ pointwise must have $a=0$, so $D'$ is tangent to the wavefronts of $D$ ($du=0$).
\end{proof}

\subsection{A family of plane-waves}\label{Example}
The purpose of this section is to construct a family $\mathcal F$ of vacuum Brinkmann spacetimes on $\mathbb M = \mathbb R\times\mathbb R\times\mathbb R^n$, parameterized by the Hilbert cube $[0,1/2]^{\mathbb Z}$, such that two elements of the sequence are diffeomorphic if and only if they are shift-equivalent under the Bernoulli shift.

Let $\mathcal P$ be the space of Brinkmann plane wave metrics on $\mathbb M$ with the standard dilation.  Then an element of $\mathcal P$ is specified by a smooth function $p\in C^\infty(\mathbb R, \mathbb R^n\odot\mathbb R^n)$ with values in the symmetric $n\times n$ real matrices.

We give the space $\mathcal P$ the weakest topology compatible with the family of seminorms:
$$\|p\|_k = \sup_{u\in[-k,k]}(|p(u)| + |\dot p(u)| +\dots + |p^{(k)}(u)|) $$
where $k$ is a positive integer, and $|\cdot|$ denotes any matrix norm.

The resulting topology is completely metrizable.  Indeed, a metric is obtained by the usual trick:
$$d(p,q) = \sum_{k=1}^\infty 2^{-k}\frac{\|p-q\|_k}{1+\|p-q\|_k}.$$
If $p_n$ is a Cauchy sequence in $\mathcal P$, then $p_n^{(k)}$ converges uniformly on every compact subset of $\mathbb R$, and so converges in the topology of $\mathcal P$ to a smooth function.

By the general arguments of \cite{panagiotopoulos2023incompleteness}, we can show the following:
\begin{theorem}
  There is no Borel measurable function $f:\mathcal P\to\mathcal X$ into a Polish space $\mathcal X$, such that $p\cong p'$ (isometry of plane waves) if and only if $f(p)=f(p')$.
\end{theorem}

We shall prove this theorem with the following construction.  Let the {\em Hilbert cube} $\mathbb H=[0,1/2]^{\mathbb Z}$ be the space of bi-infinite sequences of reals in the interval $[0,1/2]$, with the product topology.  There is then a homeomorphic action of the Bernoulli shift
$$\sigma(x)_n = x_{n+1} $$
for any $x\in \mathbb H$.  We construct a homeomorphism $\mathcal F:\mathbb H\to\mathcal P$ onto a (necessarily compact) subset, such that $\mathcal F(x)\cong\mathcal F(x')$ if and only if $x'$ is a shift of $x$. 

Fix a real number $a$ in the real interval $[0, 1)$.  Given $a$, define a function $f_a: (0, 1) \rightarrow \mathbb{R}^+$ by the formula, valid for each real number $u \in (0, 1)$:
\[ f_a(u) =  \frac{a}{u(1-u)} + \frac{(1 -a)}{u^2(1 - u)}.\]
Note that the function $f_a$ is well-defined, everywhere positive and smooth (even real  analytic).  Also we have the limits 
\[ \lim_{u \rightarrow 1^-} f_a(u) =  \infty, \quad \lim_{u \rightarrow 0^+} f_a(u) =  \infty. \]
If we write $a = 8^{-1}(9 - b^2)$, where $1 \le b \le 3$, so $b = \sqrt{9 -8a}$, then the minimum value of $f_a(u)$, for $u$ in the interval $(0, 1)$ is $\frac{(b + 3)^3}{8(b + 1)}$,  achieved when $u = \frac{b+1}{b+3}$. \\ This minimum value increases from $4$, when $a = b = 1$,  to $\frac{27}{4}$, when $a = 0$ and $b = 3$.

Define a function $g_a: \mathbb{R} \rightarrow \mathbb{R}^+$, by the formula:
\[ g_a(u) = \begin{cases} 1  &\qquad \text{if $u\ge 1$, }\\
 1  &\qquad \text{if  $u\le 0$, }\\
1 + e^{4-f_a(u)}  &\qquad \text{if  $0  < u < 1$.}
\end{cases}
\]
Then it is clear that the function $g_a: \mathbb{R} \rightarrow \mathbb{R}^+$ is well-defined and everywhere smooth and bounded, taking values in the interval $[1, 2]$.  Also $1 < g_a(u) \le 2$ in the open interval $(0, 1)$.

Now let $\alpha: \mathbb{Z} \rightarrow [0, 1)$ be a map.

For $u$ real define $p_{\alpha}(u)$ by the formula: 
\[ p_{\alpha}(u) = -1 + \Pi_{n\in \mathbb{Z}}g_{\alpha(n)}(u - n).\]
Then it is clear that the function $p_\alpha: \mathbb{R} \rightarrow \mathbb{R}$ is well-defined and smooth and 
takes values in the real interval $[0, 1]$.  Also we have the relation:
\[ p_\alpha(k) = 0 \iff k \in \mathbb{Z}.\]
We will employ this function $p_\alpha$ to define a collection of vacuum spacetimes, actually complete plane wave spacetimes, one for each map $\alpha: \mathbb{Z} \rightarrow [0, 1)$. 

\begin{theorem}\label{ShiftEquivalentBrinkmann}
Functions $\alpha,\beta:\mathbb Z\to [0,1)$ are shift-equivalent if and only if the plane-wave metrics $G_{p_\alpha}$, $G_{p_\beta}$ are equivalent under isometry.
\end{theorem}
\begin{proof}
The Theorem follows straightforwardly from Theorem \ref{BrinkmannDiffeos}.  If $\alpha,\beta$ are shift-equivalent, then $\alpha(n)=\beta(n+m)$ for some $m\in\mathbb Z$, and so $G_{p_\alpha}$ and $G_{p_\beta}$ differ by a translation $u\mapsto u+m$.
  
Conversely, suppose that $G_{p_\alpha}$ and $G_{p_\beta}$ are diffeomorphic.  We have $p_\alpha(u)$ (resp. $p_\beta(U)$) vanishes if and only if $u$ (resp. $U$) is an integer. Using Theorem \ref{BrinkmannDiffeos}, for each integer $n$:
\[ p_\beta(a(n - u_0)) = 0.\]
So $a(n - u_0)$ is integral for each integer $n$.  So $a(n - m)$ is integral for each pair of integers $m$ and $n$.  So $a$ is a non-zero integer and  $u_0 = a^{-1}p$, for some integer $p$. But we have for any integer $k$:
\[ 0 = p_\beta(k) = p_\alpha(a^{-1}(k + p)).\]
So for any integer $t$ we have that $ a^{-1} t$ is an integer.  So $a^{-1}$ is an integer.  So both $a$ and $a^{-1}$ are integral, so $a = \pm 1$ and we have now that:
\[ U = \epsilon(u - u_0), \quad p_\beta(U) = p_\alpha(u)\gamma^2, \quad \epsilon^2 = 1.\]
Since $p_\alpha$ is asymmetrical, we must have $\epsilon=+1$, so $\alpha$ and $\beta$ are shift-equivalent.
\end{proof}

We give the Hilbert cube the following metric:
$$d(\alpha,\beta) = \sum_{k=0}^\infty 2^{-k}\frac{\|\alpha-\beta\|_k}{1 + \|\alpha-\beta\|_k} $$
where $\|\alpha\|_k = \max_{i\in[-k,k]}\alpha(i).$

\begin{lemma}
  The mapping $\rho:\alpha\mapsto p_\alpha$ is a homeomorphism of $\mathbb H=[0,1/2]^{\mathbb Z}$ onto its image in $\mathcal P$.
\end{lemma}
\begin{proof}
  Fix a positive integer $k$.  For $\alpha,\beta\in [0,1/2]^{\mathbb Z}$, let
  $$d_k(\alpha,\beta) = \max_{-k\le i\le k} |\alpha(i)-\beta(i)|,$$
  and
  $$d_k'(\alpha,\beta) = \max_{-k\le i\le k}\sup_{u\in [0,1]} (|g_{\alpha(i)}(u) - g_{\beta(i)}(u)| + \cdots + |g_{\alpha(i)}^{(k)}(u)-g_{\beta(i)}^{(k)}(u)|).$$
  We then have
  $$d_k(\alpha,\beta)\le Cd_k'(\alpha,\beta)$$
  for some constant $C$, which is independent of $k$.  Indeed, for fixed $a>b$,
  $$\sup_{u\in[0,1]}(g_a(u)-g_b(u)) \ge g_a(1/2)-g_b(1/2) = \frac{e^{4a}-e^{4b}}{e^4} \ge 4e^{-4}(a-b)$$
  by the mean value theorem.  Therefore, we have
  $$d_k'(\alpha,\beta) \ge \max_{-k\le i\le k}\sup_{u\in [0,1]}|g_{\alpha(i)}(u)-g_{\beta(i)}(u)| \ge 4e^{-4}\max_{-k\le i\le k}|\alpha(i)-\beta(i)|,$$
  and so $C=e^4/4$ will do.
  
  Now let $\mathcal H\subset\mathcal P$ denote the image of $\mathbb H$ under the mapping $\rho(\alpha) = p_\alpha$.  The above argument shows that $\rho^{-1}:\mathcal H\to\mathbb H$ is continuous. Therefore to prove that $\rho$ is a homeomorphism, it is sufficient to show that $\mathcal H$ is compact in $\mathcal P$.  But $\mathcal H$ is obviously closed and it is bounded (being bounded in {\em each} of the seminorms defining the topology).\footnote{See Theorem 6.1.5 in \cite{narici2010topological}.}  A standard fact of functional analysis is that $\mathcal P$ is a nuclear space,\footnote{See Example 5.4.3 in \cite{hogbe2011nuclear}} so the Heine--Borel theorem is valid.  Therefore, the image $\mathcal H$ of $\mathbb H$ is compact in $\mathcal P$ and the mapping $\rho:\mathbb H\to\mathcal H$ is a homeomorphism.
\end{proof}

\subsection{Rosen universe isomorphisms}\label{RosenIsometries}
This section contains a description of all isomorphisms of Rosen universes.  But we first need to say what is meant by isomorphism.

\subsubsection{Isomorphisms}
In \cite{SEPI}, we introduced the class of {\em Rosen universes}, which are tensors on $\mathbb M$ having the form
$$G_\rho(h) = 2\,du\,dv - dx^Th(u)dx$$
where $h$ is a positive semidefinite matrix depending smoothly on $u$, where the set of singular points $\mathbb S$ is a discrete subset of $\mathbb U$, and each singularity is removable by a coordinate transformation.  So a Rosen universe is an ordinary plane wave, but a certain discrete set $\mathbb S$ of wave fronts where the metric degenerates in a controlled way.  Specifically, there exists a Brinkmann metric $G_\beta(p)$, globally on $\mathbb U$, such that $h=L^TL$ for some Lagrangian matrix $L$ satisfying the Jacobi equation $\ddot L + pL=0$.  The associated mapping is $\beta_L:(\mathbb M,G_\rho(h))\to(\mathbb M,G_\beta(\bar h))$, given by
\begin{equation}\label{BrinkmannizationEq}
  \beta_L(u,v,x) = \left(u,\quad v + 2^{-1}x^TL^T\dot Lx,\quad Lx\right)
\end{equation}
which satisfies $\beta_L^*G_\beta(p)=G_\rho(h)$.
\begin{definition}
  The mapping $\beta_L:\mathbb M\to\mathbb M$ given in \eqref{BrinkmannizationEq}, determined by a Lagrangian matrix $L$, is called a Brinkmannization map.
\end{definition}
  Suppose that $R$ is a constant orthogonal matrix, and denote its action on $\mathbb M$ by $\mu_R(u,v,x) = (u,v,Rx)$, which is always an isomorphism of Brinkmann spacetimes.  Note then that
$$\beta_{RL}(u,v,x) =  \left(u,\quad v + 2^{-1}x^TL^T\dot Lx,\quad RLx\right) = (\mu_R\circ \beta_L)(u,v,x).$$
Therefore,
\begin{lemma}
  If $L,\bar L$ are any two Lagrangian matrices with $L^TL=\bar L^T\bar L=h$, then the Brinkmannizations $\beta_L$ and $\beta_{\bar L}$ are isomorphic.
\end{lemma}
The following definition is therefore natural:
\begin{definition}
  Rosen universes $(\mathbb M,G_\rho(h))$ and $(\bar{\mathbb M},G_\rho(\bar h))$ are isomorphic if and only if they have isomorphic Brinkmannizations.
\end{definition}
We mention briefly that it should be possible to give local necessary and sufficient conditions, analogously to the results of \cite{SEPI}, but the formulation is a little intricate and other matters more pressing.

An isomorphism of Rosen universes is best illustrated by a diagram
\[\begin{tikzcd}
  (\mathbb M,G_\rho(h))\arrow[r,"\tilde{\phi}",dashed]\arrow[d,"\beta_L"] & (\bar{\mathbb M},G_\rho(\bar h))\arrow[d,"\beta_{\bar L}"]\\
  (\mathbb M,G_\beta(p))\arrow[r,"\phi"]&(\bar{\mathbb M},G_\beta(\bar p))
\end{tikzcd}\]
where the top dashed arrow has as domain a dense open subset of $\mathbb M$, on which it is an isometry onto a dense open subset of $\bar{\mathbb M}$.

\subsubsection{Local isometries}
We next describe the {\em local} isometries between (regular) Rosen metrics.  Recall that a Rosen metric is a metric of the form
$$G_\rho(h) = 2\,du\,dv - dx^Th(u)dx $$
where $h(u)$ is a symmetric $n\times n$ positive-definite matrix, depending smoothly on $u\in\mathbb U$ (a real interval).  Now, if $\bar h$ is defined on another interval $\bar{\mathbb U}$, we consider the problem of finding a local isometry.  We suppose, as above that $G_\rho(h)$ and $G_\rho(\bar h)$ are non-flat.  Then $G_\rho(h)$ is globally isometric to a Brinkmann space-time $G_\beta(p)$, where $p$ is the tidal curvature of $G_\rho(h)$.  Likewise, $G_\rho(h)$ is isometric to a $G_\beta(\bar p)$.  Therefore we can fill in the remaining edge in the diagram with an isometry:
\begin{equation}\label{commutativediag}
\begin{tikzcd}
  G_\rho(h) \arrow[r, dashed]\arrow[d] & G_\rho(\bar h)\ar[d]\\
  G_\beta(p) \arrow[r] & G_\beta(\bar p)
\end{tikzcd}
\end{equation}
Thus every isometry of Rosen plane waves is induced by an isometry of Brinkmann plane waves.  To describe the general isometry of Rosen, modulo the (comparatively trivial) isometries of Brinkmann, we identify the two spacetimes $G_\rho(h)$ and $G_\rho(\bar h)$, and the diagram \eqref{commutativediag} collapses to a triangle
\[
\begin{tikzcd}
 G_\rho(h) \arrow[rr, dashed]\arrow[dr] && G_\rho(\bar h)\ar[dl]\\
 & G_\beta(p)  &
\end{tikzcd}
\]
Because of the hypothesis that $G_\rho$ is non-flat, we can assume without loss of generality that all metrics are on a common domain $\mathbb U$ and the central null geodesic is pointwise fixed.  After putting the problem into this normal form, the problem is then that of characterizing maps from Brinkmann metrics to Rosen metrics, fixing the central null geodesic pointwise.  This is now easily done locally:
\begin{lemma}\label{BrinkmannizationLemma}
Given any isomorphism from a Brinkmann Penrose limit $G_\beta(p)=2\,du\,dv + x^Tpx\,du^2 - dx^Tdx$ onto a Rosen Penrose limit $G_\rho(h)=2\,du\,dV - dX^Th(u)dX$, on $M=\mathbb U\times\mathbb R\times\mathbb X$, is given by the change of variables $x=LX$, $v=V + x^TSx/2$, where $L$ is a Lagrangian matrix satisfying $\dot L = SL$ where $S$ is a solution to the Sachs equation $\dot S + S^2 + p = 0$.
\end{lemma}
\begin{proof}
  For any Lagrangian matrix $L$, the indicated substitution transforms $G_\beta(p)$ into a $G_\rho(h)$ for $h=L^TL$.
  Conversely, observe that the coordinate vectors $\partial_X$ in the Rosen coordinates are a Lagrangian basis of Jacobi fields, so the change of coordinates $L$ such that $x=LX$ is Lagrangian.
\end{proof}

\subsubsection{Global isomorphisms}
Lemma \ref{BrinkmannizationLemma} was in the case of regular Rosen metrics, but note that it is not quite satisfactory because there are more metrics $G_\rho(L^TL)$ locally then there are Rosen metrics regular on the interval $\mathbb U$, because $L^TL$ can degenerate.  The argument goes through, however, provided ``isometry of regular Rosen metrics'' is replaced by ``isomorphism of Rosen universes''.  
Because the Heisenberg isometries of a (non-flat) Rosen plane wave act transitively on the transverse null geodesics, we can select the central null geodesic of $\bar{\mathbb M}$ so that it is the image of the central null geodesic of $\mathbb M$.  Replacing $\bar{\mathbb U}$ by an affine image, the affine parameters of the central null geodesics can be made to match, and so $\bar{\mathbb U}=\mathbb U$ without loss of generality.  So, we now have two Rosen universes which are Penrose limits, and the isomorphism commutes with the dilation.  Therefore, the metrics $G_\rho(h)$ and $G_\rho(\bar h)$ map to a common $G_\beta(p)$, in a way that commutes with the dilation.

\begin{theorem}\label{RosenIsometriesTheorem}
  Let non-flat Rosen universes $G_\rho(h) = 2\,du\,dv - dx^T h(u) dx$ and $G_\rho(\bar h) = 2\,d\bar u \,d\bar v - d\bar x^T \bar h(\bar u) d\bar x$ be isomorphic.  Then after affine reparametrizations of the variable $\bar u$, and transformations by isometries belonging to the isometry group, and constant invertible coordinate changes, the isomorphism may be reduced to the transformation:
 \[ \bar u = u, \hspace{7pt} x = (I + H(u)E)\bar x,  \hspace{7pt}v = \bar v + 2^{-1}\bar x^T E H(u)E\bar x + 2^{-1}\bar x^T E \bar x. \]
 Here $H(u)$ is any (fixed) solution to the differential equation $h(u)\dot H(u) = I$ for $u\in\mathbb U-\mathbb S$, such that $(u-s)H(u)$ admits a smooth continuation at all $s\in\mathbb S$, and $E$ is a constant symmetric endomorphism.  We then have:
\[ \bar h(u) = (I + H(u)E)^Th(u)(I + H(u)E).\]
\end{theorem}
\begin{proof}
  Let $L_i=\partial_i$ be the Lagrangian basis of Jacobi fields corresponding to $G_\rho(h)$.  Consider a Lagrangian basis $\bar L_i$, such that at $u=0$, we have $L_i(0)=\bar L_i(0)$.  We have to determine the possible such $\bar L_i$.  Let $\mathbb L$ be the Lagrangian space spanned by the Jacobi fields $L_i$, and $\mathbb Q$ be the Lagrangian space of Jacobi fields vanishing at $u=0$.  Then $\mathbb L$ is spanned by $L=\partial_x$, and $\mathbb Q$ is spanned by $x^T\partial_v + \partial_x\cdot H(u)$, where $H$ is a regular primitive of $h^{-1}$.  The subspace $\bar{\mathbb L}$ spanned by $\bar L_i$ is the graph of a (constant) symmetric linear map $E:\mathbb L\to\mathbb Q$:
  $$\bar L = L + x^TE\partial_v + \partial_x\cdot HE.$$
  Putting $x=(I+HE)\bar x$ and $v=\bar v+2^{-1}\bar x^TEHE\bar x + 2^{-1}\bar x^TE\bar x$, we have
  \begin{align*}
    dx &= (I+HE)d\bar x + h^{-1}E\bar x\,du\\
    dv &= d\bar v + 2^{-1}\bar x^TEh^{-1}E\bar x\, du + \bar x^TE(I+HE)\,d\bar x.
  \end{align*}
  Therefore
  \[dx^Th\,dx = d\bar x^T(I+HE)^Th(I+HE)\,d\bar x + 2 \bar x^TE(I+HE)d\bar x\,du + \bar x^T Eh^{-1}E\bar x\,du^2\]
  \[2\,du\,dv - dx^Thdx = 2\,du\,d\bar v - d\bar x^T(I+HE)^Th(I+HE)d\bar x\]
  as required.
\end{proof}

\section{Conformal isometries}
\subsection{Brinkmann conformal automorphisms}\label{BrinkmannConformalSymmetries}

We determine the collection of conformal Killing vector fields of a given plane wave spacetime, whose metric, without loss of generality,  we take to be in the Brinkmann form:
\[ G = 2\,du\, dv + x^Tp(u)x\, du^2 - dx^T \,dx.\]
Here the spacetime manifold is $\mathbb{M} = \mathbb{U}\times \mathbb{R}\times \mathbb{X}$, with a point $X$ of $\mathbb{M}$ given by a triple $(u, v, x)$, with $u \in \mathbb{U}$,  a connected non-empty open subset of the reals, $v\in \mathbb{R}$ and $x \in \mathbb{X}$,  a Euclidean vector space of positive integral dimension $n$.

The co-ordinate differentials are written $(du, dv, dx)$, with $dx$ taking values in $\mathbb{X}$.  The corresponding partial derivatives are denoted $(\partial_u, \partial_v, \partial_x)$, with $\partial_x$ taking values in $\mathbb{X}^*$, such that the differential $d$  is $d =  du\,\partial_u +  dv \partial_v + \,dx(\partial_x)$.  When convenient, we represent the partial derivative with respect to $u$ by a dot.

\begin{lemma}\label{FundamentalLemma}
Given that  the Brinkmann plane wave $(\mathbb{M}, G)$ is conformally curved.  Let $V$ be a conformal Killing vector field, such that
$$\mathscr L_V G = SG.$$
Then
	\[ V = b\partial_v + k(2v\partial_v + x\partial_x) + x^T\dot q\partial_v + q(\partial_x) + V_{(w, W)}, \]
        \[ V_{(w, W)}=  w\partial_u + 4^{-1} \ddot w (x^Tx) \partial_v + 2^{-1}\dot wx(\partial_x) + Wx(\partial_x). \]	
Here $k$ and $b$ are real constants, whereas the endomorphism $W$ of $\mathbb{X}$ is constant and skew, representing  a (constant) element of the Lie algebra $\mathfrak{so}(\mathbb{X})$.  Also $w: \mathbb{U} \rightarrow \mathbb{R}$ is smooth.
\begin{itemize}
	\item  The conformal factor is: $S = \dot w  + 2k$
	\item  The smooth function $q: \mathbb{U} \rightarrow \mathbb{X}$ obeys the second-order homogeneous linear differential equation: \[\ddot q + pq = 0.\]
	\end{itemize}
The remaining conditions fixing the vector field $V_{(w, W)}$ are:
\begin{itemize}
	\item $0 = \dddot w + 4P\dot w + 2\dot Pw$,
	\item $0 = w\dot{\tilde p}   + 2 \dot w\tilde{p}  + \tilde{p}\omega - \omega \tilde{p}$
\end{itemize}
Here we have the smooth decomposition of $p$:
\[ p = \tilde{p} + PI, \quad  \textrm{tr}(\tilde{p}) = 0,\quad  \textrm{tr}(p) = nP.\] 
\end{lemma}
\begin{proof}

So we have spacetime dimension at least four: $n > 1$. Also $\tilde{p}(u)$ is smooth and not identically zero, so $\tilde{p}(u)$ is non-vanishing on a non-empty open subinterval $\mathbb{U}_0$ of $\mathbb{U}$ (here the interval $\mathbb{U}_0$ is not by any means necessarily unique).

Let our putative conformal Killing field be:
\[ V = A\partial_u + B\partial_v + C \partial_x.\]
Here $A$, $B$ and $C$ are smooth on the manifold $\mathbb{M}$, with $A$ and $B$ real-valued, whereas $C$ takes values in $\mathbb{X}$.

For some unknown smooth real-valued function $S$ on $\mathbb{M}$,  we write out the equation, $SG = \mathcal{L}_V G$:
\[  S\left(2\,du\, dv + (x^Tp x) du^2 - dx^T \,dx\right) =  (Ax^T\dot px + 2x^Tp C) du^2 \]
\[ +  2 dv\,dA + 2\,du\,\left( (x^Tpx)\,dA + \,dB\right) - 2\,dx^T\,dC.\]
Decomposing into components, we get six equations to analyze:
\begin{itemize}
	\item $E_{vv}$: $0 = A_v$
	\item $E_{vx}$: $0 = A_x - C^T_v$
	\item $E_{uv}$: $S = (x^Tpx)A_v + \dot A + B_v$
	\item $E_{ux}$: $0 =  - (x^Tpx) A_x - B_x + \dot C^T$
	\item $E_{uu}$: $2\dot B = - (x^T\dot px)A - 2x^Tp C - x^Tp x(2\dot A- S)$
	\item  $E_{xx}$: $SI = C_x + (C_x)^T$
\end{itemize}
The equation $E_{vv}$: $A_v =0$ gives:
\[ A = a(u, x).\]
Here $a(u, x)$ is real-valued and smooth.
Back substituting, we now have:
\begin{itemize}  \item $E_{vx}$: $0 = a_x - C^T_v$
	\item $E_{uv}$: $S = \dot a + B_v$
	\item $E_{ux}$: $0 =  - (x^Tpx) a_x - B_x + \dot C^T$
	\item $E_{uu}$: $2\dot B = - (x^T\dot px)a - 2x^Tp C - x^Tp x(2\dot a- S)$
	\item  $E_{xx}$: $SI = C_x + C_x^T$
	\end{itemize}
	Integrating the $E_{vx} $ equation gives:
	\[ C = (a_x)^T v + c(u, x). \]
	Here $c(u, x)$ is smooth and takes values in $\mathbb{X}$.
	So now we have:
	\begin{itemize}
		\item $E_{uv}$: $S = \dot a + B_v$
	\item $E_{ux}$: $0 =  - (x^Tpx) a_x - B_x + \dot a_x v + \dot c^T$
	\item $E_{uu}$: $2\dot B = - (x^T\dot px)a - 2a_xpx v - 2 c^Tpx - x^Tp x(2\dot a- S)$
	\item  $E_{xx}$: $SI = c_x + c_x^T + 2v(a_x^T)_x$
	\end{itemize}
We see that from the $E_{ux}$ and $E_{uv}$ equations:
\[ B_{xv} = \dot a_x = S_x - \dot a_x, \]
In particular, we have $(S - 2\dot a)_x = 0$ so $S  = 2\dot a + \sigma(u, v)$, where $\sigma $ is smooth and real-valued.  But from the $E_{xx}$ equation we have, in particular that  $S_{vv} = 0$, so $\sigma$ is linear in $v$.  So we have, for some smooth functions $q(u)$ and $r(u)$:
\[ S = 2\dot a + r(u) + 2q(u)v.  \]

So now the $E_{uv}$  equation  gives:
\[ B_v = S - \dot a = \dot a + r + 2q v.\]
Integrating, we get:
\[ B = \dot a(u, x)v + r(u)v + q(u)v^2 +  b(u, x). \]
Here $b(u, x)$ is real-valued and smooth.

	Back substituting, we now have:
	\begin{itemize}
	\item $E_{ux}$: $0 =  - (x^Tpx) a_x - b_x + \dot c^T$
	\item $E_{uu}$: $2(\ddot av + \dot rv + \dot qv^2 + \dot b) = - (x^T\dot px)a - 2a_xpx v - 2 c^Tpx  + x^Tp x(r + 2qv)$
	\item  $E_{xx}$: $(2\dot a + r + 2qv)I = c_x + c_x^T + 2v(a_x^T)_x$
	\end{itemize}
	The $v^2$ term in $E_{uu}$ gives $\dot q = 0$, so $q(u) = q_0$, a real constant.  The $v$ term in $E_{uu}$ gives:
	\[ \ddot a + \dot r= - a_xpx   + q_0x^Tp x\]
	The $v$ term in $E_{xx}$ gives:
	\[ q_0I = (a_x^T)_x\]
	Integrating this last equation we get:
	\[ a(u, x) = 2^{-1}q_0 x^T x + x^T y(u) + w(u). \]
	Here $w(u)$ is smooth and real-valued, whereas $y(u)$ is smooth, taking values in $\mathbb{X}$.
	The equation $\ddot a + \dot r= - a_xpx   + qx^Tp x$ now gives:
	\[  x^T\ddot y + \ddot w + \dot r = - x^Tpy. \]
	So we get:
	\[ \ddot y + py = 0, \]
	\[ r = r_0 - \dot w.\]
	Here $r_0$ is a real constant.  The remaining equations become:
\begin{itemize}
	\item $E_{ux}$: $b_x =  - (x^Tpx) (q_0 x^T  + y^T)  +  \dot c^T$
	\item $E_{uu}$: $2\dot b = - (x^T\dot px)(2^{-1}q_0 x^T x + x^Ty + w) - 2 c^Tpx  + (r_0 - \dot w)x^Tp x$
	\item  $E_{xx}$: $(2x^T\dot y + r_0  + \dot w)I = c_x + c_x^T$
	\end{itemize}
Put $c = (x^T\dot y)x - 2^{-1}\dot y (x^T x) + 2^{-1}(r_0  + \dot w)x + z(u, x)$.  Then the $E_{xx}$ equation becomes just:
	\[ 0 = z_x + z_x^T. \]
	Using abstract indices this equation becomes:
	\[ \partial_a z_b + \partial_b z_a  = 0.\]
	Here $\partial_a$ denotes the partial derivative with respect to $x^a$.
	Then we get:
        \begin{align*}
	  \partial_a \partial_c  z_b  &= - \partial_c \partial_b z_a =   \partial_b \partial_a z_c, \\
	 \partial_a (\partial_c z_b - \partial_b z_c) &= 0, \\
	 \partial_a z_b - \partial_b z_a &= 2 \omega_{ab}(u) = - 2 \omega_{ba}(u).\\
	 \partial_a z_b &= \omega_{ab}, \\
	 z^b &= {\omega_{a}}^b(u)x^a + m^b(u).\\
        \end{align*}

	So now we have:
	\[  c = \left(x^T\dot y(u) + 2^{-1}(r_0  + \dot w)\right)x - 2^{-1}\dot y(u) x^T x +  \omega(u) x + m(u).\]
	Here $\omega(u)$ is smooth and skew, taking values in $\mathbb{X}^*\otimes \mathbb{X}$, whereas $m(u)$ is smooth, taking values in $\mathbb{X}$.
	This gives:
        \begin{align*}
	 c_x &= \left(x^T\dot y(u) + 2^{-1}(r_0  + \dot w)\right)I + x\dot y^T - \dot yx^T  +  \omega, \\
	 \dot c_x &= (- x^Tpy + 2^{-1} \ddot r)I - xy^Tp + pyx^T +  \dot\omega.
        \end{align*}
        The $E_{ux}$ equation, $b_x =  - (x^Tpx) (q_0 x^T  + y^T)  +  \dot c^T$,  now gives:
\[ (b_x^T)_x =  \left(- x^Tpy +2^{-1} \ddot w - q_0x^Tpx\right)I - xy^Tp + pyx^T +  \dot\omega  - 2 (q_0 x  + y)x^Tp.  \]
Here the left-hand side is the hessian of $b$, so is symmetric, so  the right-hand side should be symmetric also.
So we get from the terms quadratic in $x$ of the righthand side, using abstract indices:
\[q_0 \left(\tilde{p}^i_m\delta^j_k +  \tilde{p}^i_k\delta^j_m -   \tilde{p}^j_m\delta^i_k -  \tilde{p}^j_k\delta^i_m\right) = 0.\] 
Taking the trace over $j$ and $k$ we get:
\[ nq_0 \tilde{p}  = 0.\]
Since $n > 1$ and $q_0$ is constant and $\tilde{p}$ is not identically zero, we get:
\[ q_0 = 0.\]
Then from the terms independent of $x$ of the right-hand side we get that the skew endomorphism  $\omega$ is constant.
Finally the linear terms give:
\[  (\tilde{p}x)\wedge  y - x \wedge (\tilde{p}y) = 0.\]   
Using abstract indices this gives:
\[ \tilde{p}^i_m  y^j -  \tilde{p}^j_m  y^i - \delta^i_m \tilde{p}^j_k y^k  + \delta^j_m \tilde{p}^i_k y^k = 0.\]
Tracing over $i$ and $m$ gives:
\[  n \tilde{p}^j_k y^k  = 0.\]
Since $n > 1$, we have $p^j_k y^k = 0$ and back substituting  gives:
\[ \tilde{p}^i_m  y^j -  \tilde{p}^j_m  y^i  = 0.\]
Contracting with $y_j$ gives the relation $(y^Ty) \tilde{p}= 0$.
But $\tilde{p}$ is non-zero on a non-empty open interval $\mathbb{U}_0\subset \mathbb{U}$.  So $y(u)$ vanishes everywhere on that open interval. But $y$ also obeys the equation $\ddot y(u) + p(u)y(u) = 0$, so $y(u)$ is identically zero, since it has trivial initial data at any point $u_0$ of $\mathbb{U}_0$, $y(u_0) = \dot y(u_0) = 0$.

So now we have:
	\[  c =  2^{-1}(r_0  + \dot w(u))x +  \omega x + m(u).\]

Then remaining equations are:
\begin{itemize}
	\item $E_{ux}$: $b_x =  2^{-1}\ddot w x^T + \dot m^T.$
	\item $E_{uu}$: $\dot b = - 2^{-1}w(x^T\dot px)   - \dot wx^Tpx -  x^Tp\omega x - x^Tpm$
	\end{itemize}
Integrating the equation for $b_x$, we get:
	\[ b = 4^{-1}\ddot w(u)x^Tx + x^T\dot m(u)+ b_0(u).\]
	Here $b_0(u)$ is smooth and real-valued.
	Then the $E_{uu}$ equation gives:
\[ 4^{-1}\dddot w(u)x^Tx + x^T\ddot m+ \dot b_0 = - 2^{-1}w(x^T\dot px)   - \dot wx^Tpx -  x^Tp\omega x - x^Tpm.	\]
So we get first that $b_0$ is a  real constant.  Then we get that $m(u)$ obeys the second-order homogeneous linear differential equation:
	\[ \ddot m + pm = 0.\]
	Finally we need the terms quadratic in $x$ to balance:
\[ 4^{-1}\dddot wx^Tx  = - 2^{-1}w(x^T\dot px)   - \dot wx^Tpx -  x^Tp\omega x \]	
Splitting this last equation into its trace-free and pure trace parts we get:
\[ \dddot w + 4\dot wP + 2w\dot P = 0, \]
\[ 0 = w\dot{\tilde p}   + 2 \dot w\tilde{p}  + \tilde{p}\omega - \omega \tilde{p}. \]
\end{proof}

Denote by $\mathcal{P}$ the $2n$-dimensional real vector space of solutions of the homogeneous linear differential equation $\ddot q(u) + p(u)q(u) = 0$, with $q(u)$ smooth and taking values in $\mathbb{X}$.  If now  $q_1$ and $q_2$ are in $\mathcal{P}$, put:
\[ (\omega(q_1, q_2))(u) = q_1(u)^T \dot q_2(u) - q_2(u)^T \dot q_1(u).\]
Then $\left(\omega(q_1, q_2)\right)^\cdot  = 0$, so $ \omega(q_1, q_2)$ is a real constant.

 Then $\omega$ gives the space $\mathcal{P}$ a natural symplectic structure.
Given a base point $u_1$ in $\mathbb{U}$, put $\mathcal{P}^+_{u_1}$ the subspace of all $q \in \mathcal{P}$ such that $\dot q(u_1) = 0$ and put $\mathcal{P}^-_{u_1}$ the subspace of all $q \in \mathcal{P}$ such that $q(u_1) = 0$.  Then each of $\mathcal{P}^\pm_{u_1}$ is a Lagrangian subspace of $\mathcal{P}$ of dimension $n$ and we have the direct sum decomposition $\mathcal{P} = \mathcal{P}^+_{u_1}\oplus \mathcal{P}^-_{u_1}$.
\begin{corollary}  Given that  the Brinkmann plane wave $(\mathbb{M}, G)$ is conformally curved, a spanning collection of conformal vector fields are the vector fields:
\begin{itemize}\item $D = 2v\partial_v + x\partial_x$, 
\item $H = \partial_v$
\item $X_q = x^T\dot q\partial_v + q(\partial_x) $, where $\ddot q + pq = 0$, 
\item  $V_{(w, W)}=  w\partial_u + 4^{-1} \ddot w (x^Tx) \partial_v + 2^{-1}\dot wx(\partial_x) + (Wx)(\partial_x)$, where the pair $(w, W)$, with $W\in \mathfrak{so}(\mathbb{X})$ constant,  obeys:
\[ 0 = \dddot w + 4P\dot w + 2\dot Pw, \quad  0 = w\dot{\tilde p}   + 2 \dot w\tilde{p}  + \tilde{p}W - W \tilde{p}.  \]
\end{itemize}
The Lie algebra commutators are as follows:
\begin{itemize}\item $ [D, H] = - 2H, \quad  [D, X_q] = - X_q, \quad  [D, V_{(w, W})] = 0$, 
\item $[H, X_q] = 0,  \quad  [H, V_{(w, W)}] = 0$,
\item $[X_q, X_r] = \omega(q, r)H, \quad   [X_q, V_{(w, W)}]  = \hspace{-1pt}  X_s, \quad  s =   \hspace{-1pt}    2^{-1}\dot wq -w\dot q + Wq$, 
\item $[V_{(w, W)}, V_{(y, Y)}]  = V_{(z, Z)}, \quad   z = w\dot y - y\dot w, \quad  Z = YW - WY$. 
\end{itemize}
\end{corollary}
Note that the operators $H$ and $X_q$ as $q$ varies, span an {\em Heisenberg Lie algebra}, over the reals, of real dimension $2n + 1$, so for $n$ degrees of freedom.  Also the operators $D$, $H$ and $X_q$ as $q$ varies, span a Lie algebra, of real dimension $2n + 2$,  which we call the {\em conformal Heisenberg Lie algebra}; this algebra is graded by the operator $D$ and has first derived algebra its Heisenberg subalgebra, spanned by $H$ (grade minus two) and the $X_q$ (grade minus one).

Multiply the equation $0 = \dot{\tilde p}w + 2\tilde{p}\dot w +\tilde{p}W - W \tilde{p}$ by $\tilde{p}$ and take the trace.  The terms involving $W$ cancel because $W$ is skew.  Also put $Q = \textrm{tr}(\tilde{p})^2 \ge 0$.  Then  $Q$ is smooth and we have:
\[ 4\dot wQ + w\dot Q = 0, \quad (w^4 Q)^\cdot = 0, \quad  w^4 Q = Q_0.\]
Here $Q_0 \ge 0$ is a real constant.   Now, by assumption,  $\tilde{p}(u)$ is non-zero for each $u \in \mathbb{U}_0$,  a non-empty open interval, whence, the function $Q$ is positive  everywhere on $\mathbb{U}_0$. 

If now the constant $Q_0$ is zero, we infer that $w$ vanishes identically on $\mathbb{U}_0$.  But $w$ obeys also the third-order linear differential equation $ 0 = \dddot w + 4P\dot w + 2\dot Pw = 0$.  Since $w$ has trivial initial data $\left(w(u_0), \dot w(u_0), \ddot w(u_0)\right) = (0, 0, 0)$ at any base point $u_0\in \mathbb{U}_0$, $w$ must vanish identically on $\mathbb{U}$.  Thus we have proved:
\begin{lemma} Either $w$ vanishes identically on $\mathbb{U}$, or $Q_0 > 0$.
\end{lemma}

Now suppose that $w$ does not vanish identically on $\mathbb{U}$. Then $Q_0 > 0$.   So $Q$ and $w^4$ are both smooth and positive everywhere on $\mathbb{U}$.  At worst after a sign change for  the pair $(w, W)$, we have that each of  $Q$ and $w$ is smooth and positive on $\mathbb{U}$.  Then by scaling the pair  $(w, W)$ by a suitable positive real constant, to make $Q_0 = 1$, without loss of generality, we may assume that:
\[ w = Q^{-\frac{1}{4}}.\]

We now substitute $w = Q^{-\frac{1}{4}}$ into the equation $\dddot w + 4P\dot w + 2\dot Pw =0$. If the resulting equation does not hold, then we are done: there are no conformal Killing vectors with $w \ne 0$.  On the other hand,  if the equation does hold, we then further substitute $w = Q^{-\frac{1}{4}}$ into the second equation:
\[ 0 = 2\tilde{p}\dot w + \dot{\tilde p}w +\tilde{p}W - W \tilde{p}.\]
If now the resulting equation holds for some constant $W = Z \in \mathfrak{so}(\mathbb{X})$ we are done and the pair $ (w, W) = (Q^{-\frac{1}{4}}, Z)$ gives a conformal Killing vector, $V_{(w, W)}$.  The general solution is then $(w, W) = (tQ^{-\frac{1}{4}}, tZ + Y)$, where $t$ is a real constant and $Y$ is a constant skew endomorphism of $\mathbb{X}$, such that the equation $pY - Y p = 0$ holds. 
\begin{definition}Denote by $\frak{W}$ the space of all conformal Killing fields of the conformally curved  Brinkmann plane wave,  $(\mathbb{M}, G)$,  and by $\frak{W}_0$ the subspace of  $\frak{W}$ consisting of all conformal Killing fields with $w = 0$.
\end{definition} 
Then we have proved:
\begin{lemma}
	 Given that  the Brinkmann plane wave,  $(\mathbb{M}, G)$,  is conformally curved, the dimension of the quotient space $\frak{W}/\frak{W}_0$ is at most one. 
\end{lemma}
\begin{definition}We say that a conformally curved Brinkmann spacetime has an {\em extra  conformal symmetry} if and only if $\frak{W}/\frak{W}_0$  has dimension one, if and only if there is a conformal Killing vector  field with everywhere non-zero component in the $u$-direction.
\end{definition}

We can be more precise. Denote by $\lambda$ the (non-negative) dimension of the space of solutions $Y$ of the equation $p(u)Y - Y p(u) = 0$, where $u$ varies over $\mathbb{U}$ and $Y \in \mathfrak{so}(\mathbb{X})$ is constant  (so $\lambda$ is the (finite) dimension of a certain proper Lie subgroup of the orthogonal group of $\mathbb{X}$).   Then we have proved:
\begin{lemma}
	 Given that  the Brinkmann plane wave,  $(\mathbb{M}, G)$,  is conformally curved,  the dimension of the algebra $\frak{W}$ of its conformal Killing vectors is either $2n + 2 + \lambda$ or $2n + 3 + \lambda$.  In the former case all conformal Killing vectors have vanishing $\partial_u$ component.  
	  In the latter case, there is an extra conformal symmetry.
	  \end{lemma}
Note that we have the commutator:
\[  [V_{(w, W)}, V_{(0, Y)}] =  V_{(0, Z)}, \quad  Z = YW - WY.\]	  
Note also that for a Killing vector field of the form $V_{(0, W)}$, the constant endomorphism $W$ must obey $Wp(u) = p(u)W$, for each $u \in \mathbb{U}$.	  

By  inspection of the various commutators we get immediately:
\begin{corollary}The first derived algebra $\mathfrak{W}_1 = [\mathfrak{W}, \mathfrak{W}]$ of the conformal symmetry algebra $\mathfrak{W}$ of a conformally curved Brinkmann spacetime  is spanned by the Heisenberg algebra, spanned by $H$ and by the $X_q$, as $q$ varies, obeying $\ddot q = pq$ and by some of the operators, $V_{(0, W)}$, where the constant skew endomorphism $W$ of $\mathbb{X}$ obeys the relation $Wp = pW$.
\end{corollary}
The commutators for the derived algebra, $\mathfrak{W}_1$, are now:
\[  [H, X_q] = 0, \quad  [X_q, X_r] = \omega(q, r)H, \]
\[  [H, V_{(0, W)}] = 0,  \quad  [X_q, V_{(0, W})]  =  X_{(Wq)},  \]
\[ [V_{(0, W)}, V_{(0, Y)}] = V_{(0, Z)}, \quad  Z = YW - WY.\]
\begin{definition}Denote by $\frak{H}$ the one-dimensional subalgebra spanned by $H$.
\end{definition}
\begin{corollary} The center of  $\mathfrak{W}_1$, the first derived algebra of the conformal Killing algebra $\mathfrak{W}$ of a conformally curved Brinkmann spacetime,  is one-dimensional, the algebra $\frak{H}$.    The centralizer in $\mathfrak{W}$ of $\frak{H}$ is spanned by $H$, by all the operators $X_q$ and by  all the operators $V_{w, W}$.
\end{corollary}
For this let $\eta = a H + X_q + V_{(0, W)}$ be a central element of $\mathfrak{W}_1$, with $a$ a real constant, $q\in \mathcal{P}$ and $W$ a constant skew endomorphism of $\mathbb{X}$. Then we have, in particular, $0 = [\eta, X_r] = \omega(q, r)H - X_{(Wr)}$, for any $r \in \mathcal{P}$.  

So first we get $Wr = 0$, for all $r \in \mathcal{P}$. But at any point $u_0$, the values of $r(u_0)$ range over all of $\mathbb{X}$, as $r$ ranges over $\mathcal{P}$ and then for any non-zero constant skew endomorphism $W$ of $\mathbb{X}$, there exists an element $r$ of $\mathbb{X}$,  such that $Wr(u_0) \ne 0$.  So $W = 0$.  Then $\omega(q, r) = 0$, for all $r\in \mathcal{P}$.  But $\omega$ is a symplectic form on $\mathcal{P}$, so $q = 0$.  So $\eta = aH \in \frak{H}$.   Conversely every element of $\frak{H}$ is clearly central, so we are done.

We next closely analyze the case with an extra conformal symmetry.  As shown above, we may take this vector field to be:
\[ V  = w\partial_u + 4^{-1} \ddot w x^Tx \partial_v + 2^{-1}\dot wx(\partial_x)  + (Wx)(\partial_x).  \]
Here we have $w$ smooth and positive everywhere on $\mathbb{U}$, $W$ a constant element of $\mathfrak{so}(\mathbb{X})$ and the following equations are satisfied:
\begin{itemize}
	\item  $0 = \dddot w + 4P\dot w + 2\dot Pw$,
	\item $0 = 2\tilde{p}\dot w + \dot{\tilde p}w +\tilde{p}W - W \tilde{p}$.
	\item $\mathcal{L}_V G = \dot w G$.
\end{itemize}
Recall the Brinkmann metric:
\[ G = 2\,du\, dv  + x^Tp(u)x\, du^2 - dx^Tdx.\]
Now put:
\[ dU = w^{-1}(u)du.\]
Since $w$ is smooth and everywhere positive,  this gives a well-defined smooth co-ordinate change.   Then we have:
\[ G = 2w du\, dv + x^Tq(U)x\, du^2 - dx^T\,dx.\]
Here $q(U) = p(u)w(u)^2$.  Put $w(u) = (z(U))^2$, for $z(U)$ smooth and everywhere positive on $\bar{\mathbb U}$, a (connected, open) smoothly diffeomorphic image of $\mathbb{U}$.  Then put $x = z(U)X$. Then,  omitting the $U$ arguments and using prime now for $U$ derivatives, we get:
\[ z^{-2}G = 2\,du\, dv + X^T(q - (\dot zz^{-1})^2 I)X\,du^2   -  dX^TdX - 2 \dot zz^{-1} X^TdX\, du\,.\]
Finally put:
\[ v = V + 2^{-1} \dot zz^{-1}X^TX.\]
Then we get:
\[ z^{-2}G = 2\,du\, dv + X^TrX\, du^2 -  dX^TdX, \]
\[ r = q + z^{-2}(z\ddot z - 2 (\dot z)^2)I.\]
So now our conformal Killing vector field $V$ has:
\[ V(U) = V du\,\frac{dU}{du} = w^{-1} V du\, = 1.\]
So in the new co-ordinate system, our conformal Killing vector has $w(U) = 1$, so may be taken to be just:
\[  V  = \partial_U + (WX)(\partial_X).  \]
Also we have, since $w(U) = 1$:
\[ 0 =  \dddot w +4P\dot w + 2\dot Pw = 2\dot P.\]
So $P$ is constant.  
Renaming appropriately,  our metric is now globally conformal to the following metric (with conformal factor $z^{-2}$):
\[ G = 2\,du\, dv + X^T p(u)X\, du^2 - dX^T dX.\]
Here we also rename $\bar{\mathbb U}$ as $\mathbb{U}$.  The (global) non-trivial conformal Killing vector $V$ may be taken to be $\partial_u + (WX)(\partial_X)$, with $W \in \mathfrak{so}(\mathbb{X})$ constant.  Finally we have, for constant $W \in \mathfrak{so}(\mathbb{X}, g)$:
 \[ 0 = \dot{\tilde p} +\tilde{p}W - W \tilde{p}.\]
 So we have now a standard Brinkmann metric:
\[ G = 2\,du\, dv + X^TR^T(u)qR(u)X\, du^2 - dX^T dX, \]
\[ p = R^T(u)qR(u), \quad  R(u) = e^{-Wu}. \]
Here $q$ is a {\em constant}  constant symmetric endomorphism of $\mathbb{X}$, not a multiple of the identity automorphism,   and $R(u)$ is a (smooth, even real analytic) {\em one-parameter subgroup} of the orthogonal group $\mathfrak{so}(\mathbb{X})$. Note that in this form this metric naturally extends as a real-analytic metric on $\mathbb{R} \times \mathbb{R} \times \mathbb{X}$.

Lastly put $x = R(u)X$ and $\dot R(u) = -WR(u)$, with $W$ skew and constant.  Then we get:
\begin{align*}
 G &= 2\,du\, dv + x^Tqx\, du^2 - d(x^TR) d(R^{-1}x)\\
  &=  2\,du\, dv + x^Tqx\,  du^2 - (dx^T - x^TW  du)(dx + Wx\, du)\\
 &= 2\,du\, dv - 2\,du\,(x^T W\,dx) + x^T(q + W^2)x \, du^2 \\
 &=  2\,du\, dv - 2\,du\,(x^T W\,dx) + x^Tpx  du^2 - dx^T dx\\
 p &= q + W^2.\\
\end{align*}
We have proved: 
\begin{theorem}
Every conformally curved Brinkmann plane wave with an {\em extra  conformal symmetry}  is conformally isometric (via a conformal factor commuting with the action of the dilations) to a real analytic metric $G_\alpha$ or to $G_\beta$,  of the following forms:
\[  G_\alpha =   du\,\alpha - dx^Tdx, \quad  \alpha = 2\,dv + (x^T e^{\omega u} q e^{-\omega u} x)  du\,.\]
\[ G_\beta =   du\,\beta - dx^Tdx, \quad  \beta = 2\,dv - 2x^T\omega dx + (x^Tp x)  du\,.\]
Here each of  $\omega$,  $p$ and $q$ is a  {\em constant  endomorphism} of the Euclidean space $\mathbb{X}$.  Also $\omega$ is skew, whereas each of  $p$ and $q$ is symmetric.  Finally $q = p - \omega^2$ is not a multiple of the identity.
\end{theorem}
Note that, a priori, the given plane wave may not use the whole $u$-range, but the metric $G$,  just given,  allows $u$ to be any real.

The conformal symmetries $\mathfrak{W}$ of $G_\beta$ are as follows:
\[ J = \partial_u, \quad H = \partial_v,  \quad D = 2v\partial_v + x\partial_x,  \quad X_{B} = x^T(\dot B - \omega B)\partial_v + B \partial_x,  \quad L_{W} =  (Wx)(\partial_x).\]
Here $B(u)$ is real analytic,  taking values in $\mathbb{X}$ and we have:
\[ \ddot B - 2\omega \dot B +  pB = 0. \]
Also $W$ is a constant skew endomorphism of $\mathbb{X}$,  so a constant element of $\mathfrak{so}(\mathbb{X})$, which must preserve both $p$ and $\omega$.
\[ pW = Wp, \quad  \omega W = W\omega.\]
Note that in this form it is natural to allow $u$ to range over the entire real line, with $B$ real analytic on $\mathbb{R}$.  {\em Then the conformal symmetry group is transitive}.  Also in this form all symmetries are Killing symmetries except for the dilations. 

 The Lie algebra commutators are:
\[ [J, D] = 0, \quad [J, H] = 0,\quad  [J, X_B] = X_{\dot B},  \quad [J, L_W] = 0, \]
\[ [D, H] = - 2H, \quad [D, X_B] = - X_B, \quad [D, L_W] = 0, \]
\[ [H, X_B] = 0, \quad [H, L_W]   = 0, \]
\[ [X_{B}, X_{C}] = \mu(B, C) H, \quad [X_B, L_W] = X_{WB}, \]
\[ [L_W, L_Y] = L_Z, \quad  Z = YW - WY.\] 
Here $W$ and $Y$ are skew endomorphisms of $\mathbb{X}$ that commute with $p$ and $\omega$.  Also $B$ and $C$ are real analytic in $u$,  taking values in $\mathbb{X}$, such that $\ddot B - 2\omega \dot B +  pB = 0$ and $\ddot C - 2\omega \dot C +  pC = 0$. Finally $\mu(B, C) = B^T \dot C - C^T \dot B - B^T\omega C + C^T \omega B$ is a real constant skew bilinear over the reals in $B$ and $C$.
By inspection, we see that the first derived algebra $\mathfrak{W}_1$ is spanned by $H$, by all $X_B$ and by the first derived algebra of the subalgebra of $\mathfrak{so}(\mathbb{X})$, that preserves both $p$ and $\omega$.  In particular, its center is spanned by $H$.

Finally the isotropy algebra of the conformal symmetry group at the point $(u, v, x) = (u_0, 0, 0)$ on the central null geodesic is spanned by $D$ and by all $L_W$ and by all $X_B$, such that $B(u_0) = 0$.
\begin{definition} A conformally curved plane wave spacetime with an extra conformal symmetry, whose conformal symmetry group is {\em  transitive}, is called a {\em microcosm}.
\end{definition}
\begin{corollary} The general microcosm, $(\mathbb{M}, G)$ may be written, on the manifold $\mathbb{M} = \mathbb{R} \times \mathbb{R} \times \mathbb{X}$, where the real Euclidean vector space $\mathbb{X}$ has dimension $n > 1$, with metric $G$ given as follows, up to a positive conformal factor:
\[ G = 2\,du\,\beta - dx^T \,dx, \]
\[ \beta = 2 dv - 2x^T \omega \,dx + (x^T px)  du\,.\]
Here $\omega$ and $p$ are endomorphisms of $\mathbb{X}$, with $\omega$ skew and $p$ symmetric. Also $p - \omega^2$ is not pure trace.
\end{corollary}

\subsection{Brinkmann conformal isometries}

We now prove:
\begin{theorem}
  Let $(\mathbb{M}, G_\beta(p))$ and $(\bar{\mathbb M}, G_\beta(\bar p))$ be conformally isometric Brinkmann plane waves, and suppose that $G_\beta(p)$ is conformally curved.  Then the conformal isometry $\phi:\mathbb M\to\bar{\mathbb M}$ can be factored as
  $$\phi = \alpha\circ\rho\circ\mu $$
  where:
  \begin{itemize}
  \item $\mu:\mathbb M\to\bar{\mathbb M}$ is a mapping
    $$\mu(u,v,x) = \left(U(u),\quad v\op{sgn}\dot U + x^Tx\frac{\ddot U}{4|\dot U|},\quad \frac{x}{|\dot U|^{1/2}}\right)$$
    induced by a reparametrization $U(u)$ of the coordinate $u$;
  \item $\rho$ belongs to the isometry group of $G_\beta(\bar p)$;
  \item $\alpha$ is the constant transformation, linear on each wave front
    $$\alpha(u,v,x) = (u,\ a^2v,\ aAx)$$
    where $a\not=0$, $a\in\mathbb R$, and $A\in\op{O}(\mathbb X)$.
  \end{itemize}
  \end{theorem}
\begin{proof}

  Because both metrics are conformally curved, the center of the first derived algebra of their respective conformal Killing algebras is one-dimensional, spanned by $\partial_v$ and $\partial_{\bar v}$, so that $d\phi(\partial_v) = a\partial_{\bar v}$ for some constant $a\not=0$.  Therefore, $\phi$ maps each null hypersurface orthogonal to $\partial_v$ onto one orthogonal to $\partial_{\bar v}$.  Therefore, $\phi$ sends the distribution $du=0$ onto $d\bar u=0$.  That is, $\bar u$ is a non-singular function of $u$.  Moreover, the vector field $\partial_v$ maps to a constant multiple of $\partial_V$, we can assume without loss of generality, after absorbing a positive conformal factor and rescaling the pair $(v,x)$ to $(c^2v,cx)$ for some constant real $c\not=0$.

  We first deal with the $\mu$ part of the conformal isometry.  By replacing $\bar u$ by $-\bar u$ and $\bar v$ by $-\bar v$ if necessary, we can assume that $\dot{\bar u} >0$.  Let $L^2(\bar u)\,du = d\bar u$.  Then, denoting the derivative with respect to $U=\bar u$ with a prime, the diffeomorphism
  $$U=U(u), \quad x=L(U)^{-1}X, \quad v = V-2^{-1}X^TX L^{-1}L'$$
  is an isometry between $L(U)^2G_\beta(p(u))$ and $G_\beta(P(U))$ where
  $$P(U) = L(U)^{-4}p(u(U)) - L''(U)/L(U).$$
  That is,
  $$L(U)^2(2\,du\,dv + x^Tp(u)x\,du^2 - dx^Tdx) = 2\,dU\,dV + X^T P(U)X\,dU^2 - dX^TdX.$$
  So we can take $\bar u = u$ after absorbing a conformal factor depending only on $u$.

  The mapping $\phi$ maps the central null geodesic of $G_\beta(p)$ to a transverse null geodesic for $G_\beta(\bar p)$.  Applying a Heisenberg group symmetry of $G_\beta(\bar p)$, we can assume without loss of generality that $\phi$ maps the central null geodesic to itself.

  We now describe the remaining freedom in the diffeomorphism, which we write
  \[ (\bar u, \bar v, \bar x) = \Phi(u, v, x) = \left(u, V(u, v, x), X(u, v, x)\right).\]
  Because the first derived algebra of the conformal algebra is invariant under the conformal isometry, not only is the $u$ coordinate preserved, but the affine structure on the wave fronts is preserved.  Therefore, at fixed $u$, $V$ and $X$ are affine functions, which are thus also linear because they preserve the central null geodesic.  Also,
  $$\partial_v = \bar \partial_{\bar v} = V_v\bar \partial_{\bar v} + X_v(\bar \partial_{\bar x})$$
  which implies that $X_v = 0$ and $V_v = 1$.  So the transformation simplifies to:
  \[V = v + x^Tw(u), \qquad X = L(u)x\]
  where $w,L$ are smooth in $u$ and $L$ is an invertible matrix.
  Substituting gives $\dot w=0$ and
  \begin{align*}
    2\,du\,dV &+ X^TP(U)X\,du^2 - dX^TdX \\
              &= 2\,du\,dv + X^TQ(U)X\,du^2 + 2(w^T - x^TL^T\dot L)\,dx\,du - dx^TL^TLdx\\
    Q(U)      &= P(U) - L^{-T}\dot L^T\dot LL^{-1}
  \end{align*}
  From the $dx\,du$ term, $w=0, \dot L=0$.  From the $dx^Tdx$ term, $L^TL=I$.
\end{proof}

As a corollary, Theorem \ref{ShiftEquivalentBrinkmann} holds if ``isometry'' is replaced by ``conformal isometry'':
\begin{corollary}
  In the construction of \S\ref{Example}, functions $\alpha,\beta:\mathbb Z\to [0,1)$ are shift-equivalent if and only if the plane-wave metrics $G_{p_\alpha}$, $G_{p_\beta}$ are equivalent under conformal diffeomorphism.
\end{corollary}

\subsection{Rosen universe conformal isomorphisms}\label{ConformalDiffeoRosen}
This section contains a description of all conformal isometries of Rosen universes.  It is largely parallel with \S\ref{RosenIsometries} above.

\begin{definition}
  Let $(\mathbb M,G_\rho(h))$ and $(\bar{\mathbb M},G_\rho(\bar h))$ be Rosen universes, where $\mathbb M=\mathbb U\times\mathbb R\times\mathbb X$, $\bar{\mathbb M}=\bar{\mathbb U}\times\mathbb R\times\mathbb X$, and $h$ is positive definite away from the discrete set $\mathbb S\subset\mathbb U$, and $\bar h$ is positive definite away from the discrete set $\bar{\mathbb S}\subset\bar{\mathbb U}$.  A conformal isomorphism of Rosen universes is a smooth function
  $$\phi : \mathbb M-\mathbb S\to\bar{\mathbb M} $$
  that restricts to a conformal isometry from a dense open subset of $\mathbb M$ onto a dense open subset of $\bar{\mathbb M}$, which sends every null geodesic to a null geodesic.
\end{definition}

The main theorem is almost the same as Theorem \ref{RosenIsometriesTheorem}:
\begin{theorem}\label{RosenConfIsometriesTheorem}
  Let conformally curved Rosen universes $G_\rho(h) = 2\,du\,dv - dx^T h(u) dx$ and $G_\rho(\bar h) = 2\,d\bar u \,d\bar v - d\bar x^T \bar h(\bar u) d\bar x$ be conformally isomorphic.  Then after reparametrizations of the variable $\bar u$, and transformations by isometries belonging to the first derived conformal symmetry group $\mathcal W_1$, and constant invertible coordinate changes, the isomorphism may be reduced to the transformation:
 \[ \bar u = u, \hspace{7pt} x = (I + H(u)E)\bar x,  \hspace{7pt}v = \bar v + 2^{-1}\bar x^T E H(u)E\bar x + 2^{-1}\bar x^T E \bar x. \]
 Here $h(u)\dot H(u) = I$ for $u\in\mathbb U-\mathbb S$, such that $(u-s)H(u)$ admits a smooth continuation at all $s\in\mathbb S$, satisfying the initial condition that $H(0)=0$, and $E$ is a constant symmetric endomorphism.  We then have:
\[ \bar h(u) = (I + H(u)E)^Th(u)(I + H(u)E).\]
\end{theorem}

The strategy of proof is the same as Theorem \ref{RosenIsometriesTheorem}.  Given a conformal isomorphism of Rosen universes $G_\rho(h)\to G_\rho(\bar h)$, we isometrically map $G_\rho(h)$ and $G_\rho(\bar h)$ to $G_\beta(p)$ and $G_\beta(\bar p)$, respectively:
\begin{equation}\label{commutativediag1}
\begin{tikzcd}
  G_\rho(h) \arrow[r, dashed]\arrow[d] & G_\rho(\bar h)\ar[d]\\
  G_\beta(p) \arrow[r] & G_\beta(\bar p)
\end{tikzcd}
\end{equation}
Thus every conformal isometry of Rosen universes is induced by a conformal isometry of Brinkmann plane waves.  Again, identifying the Brinkmann plane waves at the bottom, the diagram \eqref{commutativediag1} collapses to a triangle
\[
\begin{tikzcd}
 G_\rho(h) \arrow[rr, dashed]\arrow[dr] && G_\rho(\bar h)\ar[dl]\\
 & G_\beta(p)  &
\end{tikzcd}
\]
where all arrows are now {\em isometries} on dense open sets.

Thus, modulo a conformal isometry of the Brinkmann metrics $G_\beta(p)$ and $G_\beta(\bar p)$, Theorem \ref{RosenConfIsometriesTheorem} reduces to the statement of Theorem \ref{RosenIsometriesTheorem}.

\bibliographystyle{hplain} 
\bibliography{planewaves} 

\end{document}